\documentclass[11pt,fleqn]{article}
\usepackage{amssymb,latexsym,amsmath,amsfonts,graphicx,amsthm}
\usepackage{pictex,color,comment}
\usepackage[breaklinks=true]{hyperref}

    \advance\textheight by \topskip
    \numberwithin{equation}{section}

\def\beq{\begin{equation}}
\def\eeq{\end{equation}}
\def\bdf{\begin{definition}}
\def\edf{\end{definition}}
\def\nn{\nonumber}
\def\bit{\begin{itemize}}
\def\eit{\end{itemize}}

\makeatletter
\def\eqalign#1{\null\vcenter{\def\\{\cr}\openup\jot\m@th
  \ialign{\strut$\displaystyle{##}$\hfil&$\displaystyle{{}##}$\hfil
      \crcr#1\crcr}}\,}
\makeatother

\newcommand{\C}[1]{{\mathcal{#1}}}

\newcommand{\be}{\begin{equation}}
\newcommand{\ee}{\end{equation}}

    \def\e{{\epsilon}}

    \def\Re{{\rm Re \,}}
    \def\Im{{\rm Im \,}}

    \def\bigO{{\cal O}}

    \def\P2n{{\rm P}_{{\rm II}}^{(n)}}

    \newtheorem{theorem}{Theorem}[section]
    \newtheorem{lemma}[theorem]{Lemma}

    \newtheorem{Definition}[theorem]{Definition}
    \newenvironment{definition}{\begin{Definition}\rm}{\end{Definition}}
    \newtheorem{Remark}[theorem]{Remark}
    \newenvironment{remark}{\begin{Remark}\rm}{\end{Remark}}
    \newtheorem{Example}[theorem]{Example}
    
    \newtheorem{Assumptions}[theorem]{Assumptions}
    \newenvironment{assumptions}{\begin{Assumptions}\rm}{\end{Assumptions}}

    \newcommand{\supp}{{\operatorname{supp\,}}}

    \DeclareMathOperator*{\Tr}{Tr}
    
\begin{document}
\title{A Riemann-Hilbert problem for equations of Painlev\'{e} type in the one matrix model with semi-classical potential}
\author{Max R. Atkin}
\maketitle

\begin{abstract}
We study the hermitian one matrix model with semi-classical potential. This is a general unitary invariant random matrix ensemble in which the potential has a derivative that is a rational function and the measure is supported on some collection of disjoint closed intervals. Such models have attracted much interest both due to their physical applications and relations to integrable systems. An object of central interest in random matrix theory is the correlation kernel, as this encodes the eigenvalue correlation functions. In recent years many results have been obtained proving that the correlation kernel near special points in the spectrum can be expressed in terms of Painlev\'{e} transcendents and their associated Riemann-Hilbert problems. In the present work we build on this success by proposing a model problem that is general enough to describe the limiting kernel at any point in the spectrum. In the most general situation this would include cases of logarithmic singularities and essential singularities in the weight colliding with soft or hard edges, the bulk of the spectrum or even births of a cut. 
\end{abstract}

\newpage
\tableofcontents
\newpage

\section{Introduction}
Let $\C{I}$ be a finite collection of closed pairwise disjoint intervals in $\mathbb{R}$ and $V:\mathbb{R}\rightarrow \mathbb{R}$ a function whose derivative is rational. Then define a {\emph{semi-classical hermitian matrix model}} as a random $n \times n$ hermitian matrix given by $M = U^\dagger \mathrm{diag}(x_1,\ldots,x_n) U$ where $U$ is a Haar distributed unitary matrix and the eigenvalues $x_i$ are distributed according to the probability measure,
\beq
\label{semiclassical mu}
\frac{1}{Z_n} \Delta(x)^2 \prod^n _{i=1}e^{-n V(x_i)} \chi_{\C{I}}(x_i) dx_i.
\eeq
Here $\chi_{\C{I}}$ is the indicator function on $\C{I}$, $\Delta(x)$ is the Vandermonde determinant, the function $V$ is known as the {\em{potential}} and we often refer to $w(x) := e^{-n V(x)}$ as the weight. The normalisation constant $Z_n$,
\beq
Z_n = \int_{\C{I}^n} \Delta(x)^2 \prod^n _{i=1}e^{-n V(x_i)} \chi_{\C{I}}(x_i) dx_i.
\eeq
is known as the {\em{partition function}}. Note that implicit in this definition is that $V$ is such that $Z_n$ exists.

The eigenvalues distributed according to \eqref{semiclassical mu} form a determinantal point process. The correlation functions for this process can therefore all be expressed in terms of a correlation kernel. In turn the correlation kernel is expressed using orthogonal polynomials with respect to the weight $w(x)$ on $\C{I}$. More precisely, let $p_j$, $j=0,1,\ldots$ be the family of monic polynomials of degree $j$ characterised by the relations
\beq\label{ortho p}
\int_{\C{I}} p_j(x) p_m(x) w(x) dx = h_j \delta_{j m}.
\eeq
The correlation kernel can now be written as
\beq
\label{Keqn}
K_n(x,y) = h_{n-1}^{-1}\frac{\sqrt{w(x)w(y)}}{x-y} \left(p_n(x)p_{n-1}(y) - p_n(y)p_{n-1}(x) \right).
\eeq
The asymptotic analysis of the correlation kernel therefore reduces to the asymptotic analysis of the orthogonal polynomials $p_n$. Indeed, the use of the term semi-classical to describe such an ensemble arises from the fact that their associated orthogonal polynomials are of semi-classical type.

In the present work we will obtain $n \rightarrow \infty$ asymptotics for the eigenvalue correlation kernel of such a model under very general conditions. In particular we obtain the asymptotic behaviour of the kernel at points in which edges of $\C{I}$ and singularities in the potential approach the support of the eigenvalues as $n \rightarrow \infty$ in a way that ensures a non-trivial double scaling limit. We will express the asymptotics for the kernel in all such cases in terms of a general model Riemann-Hilbert problem which generalises nearly all other model problems previously introduced in the context of random matrix theory. Some examples of the scenarios not addressed in the analysis here include cases when the potential has isolated discontinuities, which have been studied in \cite{jumpHypergeometric} and lead to confluent hypergeometric kernels and also the interpolating kernel found in the context of the birth-of-a-cut \cite{ClaeysBirthofCut}. Finally, for reasons of space we restrict ourselves to the case of logarithmic singularities in $V$ with positive coefficients.

\subsection{Motivations}
We begin this section with a historical review of the physical applications of random matrix theory, after which we turn to the mathematical motivations for the current study.

The one hermitian matrix model in which $V$ is polynomial and $I = \mathbb{R}$ was the first random matrix model to be extensively studied. These studies culminated in significant applications to problems in particle physics \cite{GMreview}. Further applications quickly followed, with many requiring some modification of the model to allow $V'$ to be rational or to allow $I \neq \mathbb{R}$. To be more precise it was found that for applications of random matrix theory to QCD \cite{RMTQCDreview}, quantum transport problems \cite{BFB1,BFB2} and string theory \cite{FZZTbranesBanksDouglas}, the models necessarily included logarithmic divergences in $V$.

Since that time there has been an explosion in the number of applications of random matrix theory and many of the proposed models fall into the class of semi-classical hermitian matrix models. This is especially true in the context of quantum transport in which recent models incorporating poles in to the potential have been shown to have physical relevance. There the observable of interest is the Wigner-Smith time-delay matrix $Q$, whose eigenvalues $\tau_j$ are related to the ``inverse delay times'', $\gamma_j = \tau_j^{-1}$. The joint probability density for the $\gamma_j$ was first obtained in \cite{BFB1,BFB2} where it was shown it took the form,
\beq
P(\gamma_1, \ldots, \gamma_n) =\frac{1}{Z_{n}} \left |\Delta(\gamma)\right |^\beta \prod_{j=1}^n \gamma_j^{\beta n/2} e^{-\frac{\beta}{2} \gamma_j},
\eeq
where $\beta$ depends on the symmetries of the system, with $\beta=2$ a common case. Since many observables may be expressed in terms of $Q$, the problem of computing expectation values with respect to the above measure is relevant. Let us highlight the recent work \cite{MT} and \cite{GT} in which physically important observables were expressed in terms of $\Tr Q$ and $\Tr Q^2$. The associated moment generating function for these observables is the partition function for a semi-classical hermitian matrix model in which the potential contains a singularity of order $k=1$ and $k=2$ respectively.

Random matrix theory has also found applications in integrable quantum field theory at finite temperature. Here again it was found that models of semi-classical type were relevant \cite{CI}.

Finally, another important application of random matrix theory is in the field of analytic number theory. Here random matrix theory serves as a source of models for the zeros of various $\zeta$-like functions. An example of this relation was given recently by \cite{FyodorovKeating} in which the statistics of large deviations of $|\zeta(1/2+ix)|$ in an interval on the critical line was found to be related to the statistics of large deviations of $|\det(U-e^{i\theta}I)|$, where $U$ is a $n\times n$ Harr-distributed unitary matrix.

Of central importance in the work \cite{FyodorovKeating2} was a conjecture for the asymptotics of certain expectation values using a measure of the form \eqref{semiclassical mu} but where the measure is supported on a circle in the complex plane rather than the real line. Such ensembles are known as Circular Unitary Ensembles (CUE). The CUE is very closely related to the semi-classical hermitian matrix model and we expect the same set of model problems as introduced here to appear in the CUE context. To address the conjecture made in \cite{FyodorovKeating2} would require analysing model problems with an arbitrary number of logarithmic singularities, which is a property of the model problems introduced in this paper. This is something we hope to pursue in the near future.

Another work related to the $\zeta$ function is \cite{BMM}. There the motivating interest was in a random matrix model describing a certain observable with relevance to the Riemann hypothesis. The authors of \cite{BMM} studied a semi-classical hermitian model which contains a simple and second order pole .

We now turn our attention to motivations arising from Riemann-Hilbert analysis. The use of Riemann-Hilbert problems and the associated machinery of the Deift-Zhou steepest descent has enjoyed considerable success in random matrix theory over the last fifteen years. It has been able to give rigorous proofs of many conjectures made in the physics literature. We highlight especially the strong universality results that have been obtained using these methods. For example it has established universality of the Sine, Airy, Bessel and numerous Painlev\'{e} kernels, for a broad class of potentials.

The industry of applying the Deift-Zhou steepest descent to the RH problems arising in the context of random matrix theory is still strong with numerous results in the last few years \cite{BMM,XDZ,XDZ2,ACM}. These results in particular have begun to address the types of new critical behaviour that one finds in the semi-classical hermitian matrix model. The current work was motivated by the desire to continue this program. In doing so we obtain proofs for many small conjectures in the Riemann-Hilbert analysis literature. For instance our main results prove the conjectures made in \cite{ClaeysBirthofCut} and \cite{CIK} concerning the most general behaviour at the birth of a cut and at a critical edge point.

Finally let us note that there exists an extensive literature \cite{TRreview} on an alternative method for computing asymptotics of random matrix integrals, known as topological recursion. Topological recursion gives a uniform structure to the asymptotic expansion of observables in random matrix theory. Furthermore, the recursion needs only the data of the spectral curve together with a geometric object known as the Bergmann-kernel. This work was partly motivated by a desire to better understand the role of these objects in the Riemann-Hilbert approach. In particular our results show that the spectral curve directly determines the double scaling limit in all cases, without the need for the kind of smart approximations made in \cite{claeysPII2005, claeys2008multi, ClaeysBirthofCut}. The Bergmann-kernel is a bit more mysterious in the current formulation but this information likely resides in the global parametrix together with the RH problem for the orthogonal polynomials.

\subsection{Statement of results}
Our results are the following:
\begin{itemize}
\item[1.] We introduce a general model RH problem which we refer to as the {\em{canonical model problem}}. This model problem generalises many of the model problems already used in the random matrix theory literature, such as the Airy, Bessel, Painlev\'{e} I and II model problems. The canonical model problem has jump matrices independent of the parameters of the system and therefore possesses an associated system of linear ODEs which corresponds, when it exists, to the Lax pair of a Painlev\'{e} type equation.
\item[2.] We prove the existence of a solution to the canonical model problem for a certain range of model parameters.
\item[3.] We prove, given a suitable definition of the double scaling limit, that the double scaled correlation kernel of the eigenvalues may always be expressed in terms of a solution to the canonical model problem.
\end{itemize}

We will state our results more precisely in Theorem \ref{existence thm} and Theorem \ref{Kthm}.
This will require that we first introduce some more notation and definitions in order to define
the canonical model problem and the double scaling limit.


\subsubsection{The canonical RH problem for $\Phi$}
Let us begin by motivating the following definitions. The model problem will arise as the local behaviour of a RH problem for orthogonal polynomials associated with the weight in \eqref{semiclassical mu}. Naturally therefore, the model problem must contain the features from the weight in \eqref{semiclassical mu}. In particular it must contain a collection of closed disjoint intervals $I$ which encodes the local behaviour of $\C{I}$ together with a set $B$ which represents the singular points of $V$ and a set of vectors $\{\vec{\tau}_b \}$ giving the coefficients of the singular terms in $V$. With this in mind we now give some definitions.

\begin{figure}[t]
\centering 
\includegraphics[scale=0.4]{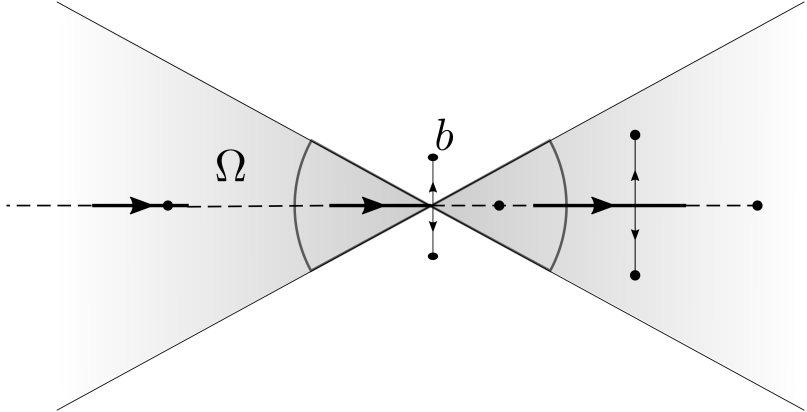}
\caption{An example of a system of contours in a canonical model problem. The thick lines are the intervals in $I$. The thin dashed lines are the intervals in $\bigcup_{b \in B} \Gamma_b \setminus I$. The black dots, one of which is labelled $b$, are locations of points in $B$ and the contours ending at these points are the $\Gamma_b$ contours. The shaded grey regions denoted by $\Omega$ correspond to the potential location of sectors in which the asymptotic behaviour of $\Phi$ differs.}
\label{modelcontoursNoLens}
\end{figure}

\bdf
Let $B$ be a finite set of points $b \in \mathbb{C}$ together with their complex conjugates. To each $b \in B$ we associate a vector $\vec{\tau}_b$ of length $\hat{d}_b$.
\edf
Because the singular points in $b$ may contain logarithmic branch points we must also define the jump contours ending at each $b$.
\bdf
Given a point $b \in \mathbb{C}$, define the contour,
\beq
\Gamma_{b} := (-\infty, \Re b) \cup (\Re b + i[0,\Im b])
\eeq
with the orientation of the contour directed away from $-\infty$.
\edf

\bdf
\label{zalphabr}
Let $[z-b]^{2\alpha}$ denote the function $(z-b)^{2\alpha}$ in which the branch cut is taken on $\Gamma_b$ and $[z-b]^{2\alpha}$ is positive for $z-b \in \mathbb{R}^+$.
\edf

We need a function on $\bigcup_{b \in B} \Gamma_{b}$ to describe the jumps due to the logarithmic divergences. We give this in the following definition together with an extension of the definition into the complex plane. For convenience define $\hat{\alpha}_b := \tau_{b,0} \in [0,\infty) $ for each $b \in B$. Then introduce,
\bdf
\beq
\label{hat alpha def}
\hat{\alpha}_{\Gamma}(z) := \begin{cases} \sum_{b \in B} \chi_{\Gamma_{b}}(z) \hat{\alpha}_{b}, & z \in  \bigcup_{b \in B} \Gamma_{b}\\ \sum_{b \in B} \chi_{\Gamma_{b}}(\Re(z)) \hat{\alpha}_{b}, & z \in  \mathbb{C} \setminus \bigcup_{b \in B} (\Re b + i\mathbb{R}).
\end{cases}
\eeq
\edf
Finally we define the very useful function,
\bdf
\label{theta def} $\theta(z) := \pm 1$ for $\pm \Im z> 0.$
\edf
We are now in a position to define the relevant class of model problems.

\bdf Let $I$ be a collection of pairwise disjoint closed, possibly infinite, intervals. Let $\vec{\tau}_\infty$ be a vector whose length we will specify later. We define a {\em{canonical model problem}} for a function $\Phi_k(z|I,B,\{\vec{\tau}_b\},\vec{\tau}_\infty)$ to be the following RH problem. Note that in the following we suppress any unnecessary arguments of $\Phi$.

\subsubsection*{RH problem for $\Phi$}
\begin{itemize}
\item[(a)] $\Phi:\mathbb{C} \setminus I \setminus \cup_{b \in B} \Gamma_{b} \rightarrow \mathbb{C}^{2\times2}$ is analytic in $z$. See Figure \ref{modelcontoursNoLens}. 
\item[(b)] The jump matrix $j(z) :=\Phi_-(z)^{-1}\Phi_+(z) $ has the following form,
\begin{align}
j(z) &= \left(
\begin{array}{cc}
 e^{2 \pi i  \hat{\alpha}_\Gamma(z)} & \chi_I(z) \\
 0 & e^{-2 \pi i  \hat{\alpha}_\Gamma(z)} \\
\end{array}
\right), &z\in I \cup \bigcup_{b \in B} \Gamma_{b}.
\end{align}
\item[(c)] To describe the asymptotic behaviour of $\Phi$ as $z\to\infty$ we need the function,
\beq
\hat{K}(z; \Omega) := \begin{cases}
\begin{pmatrix}1&0\\-\theta(z)e^{2\pi i \hat{\alpha}_\Gamma \theta(z)}&1\end{pmatrix}, & \mbox{for } z \in \Omega \\
I, & \mbox{for } z \in \mathbb{C} \setminus \Omega.
\end{cases}
\eeq
where $\Omega$ is a sector in $\mathbb{C}$ and $\theta$ is defined in Definition \ref{theta def}.

Define a polynomial of degree $k$ by,
\beq
\hat{P}_{k}(z) :=\sum^{k}_{j=0} \tau_{\infty,j} z^j.
\eeq
We then have three possibilities as $z\rightarrow \infty$,
\begin{itemize}
\item[(i)] we say $\Phi_k$ is {\em{exterior-type}} of order-$k$ with behaviour,
\beq
\label{Phi z to inf ex}
\Phi_k(z)=\left(I+\bigO(z^{-1})\right) z^{(\hat{\alpha}_\mathrm{tot} + \hat{c} ) \sigma_3} e^{-\frac{1}{2k} \hat{P}_{2k}(z) \sigma_3},
\eeq
where $\hat{c} \in \mathbb{N}$, $k \in \mathbb{N}$ and $\hat{\alpha}_\mathrm{tot} = \sum_{b\in B} \hat{\alpha}_b$.

\item[(ii)] we say $\Phi_k$ is {\em{edge-type}} of order-$k$ with behaviour,
\beq
\label{Phi z to inf edge}
\Phi_k(z)=\left(I+\bigO(z^{-1})\right) z^{-\frac{\sigma_3}{4}} N e^{-\frac{2}{2k+3} z^\frac{1}{2}\hat{P}_{k+1}(z) \sigma_3}\hat{K}(z; \Omega^{(\mathrm{edge})})^{-1},
\eeq
where $N=\frac{1}{\sqrt{2}}(I+i\sigma_1)$, $k \in \{-1\} \cup 2\mathbb{N}^0$ and $\tau_{\infty,k+1}  \in \mathbb{R}^+$ for $k \in 2 \mathbb{N}^0$ and $\tau_{\infty,k+1} \in \mathbb{R}^-$ for $k = -1$. We also introduced the sector $\Omega^{(\mathrm{edge})} = \{z : \pi > |\arg z| > \pi -  \pi/(2k+3)\}$ if $k \in 2\mathbb{N}^0$ and $\Omega^{(\mathrm{edge})} = \{z : \pi > |\arg z| > \pi -  \pi/3\}$ when $k = -1$.

\item[(iii)] we say $\Phi_k$ is {\em{interior-type}} of order-$k$ with behaviour,
\beq
\label{Phi z to inf int}
\Phi_k(z)=\left(I+\bigO(z^{-1})\right) Q(z) e^{-\frac{i}{2k+1}\hat{P}_{2k+1}(z) \theta(z) \sigma_3}\hat{K}(z;\Omega^{(\mathrm{int},-)}\cup \Omega^{(\mathrm{int},+)})^{-1}.
\eeq
where $k \in \mathbb{N}^0$, $\tau_{\infty,2k+1} \in \mathbb{R}^+$,  $\Omega^{(\mathrm{int},+)} =  \{z : |\arg z| < \pi/(4k+2)\}$, $\Omega^{(\mathrm{int},-)} =  \{z : \pi > |\arg z| > \pi - \pi/(4k+2)\}$, $\theta$ is defined in Definition \ref{theta def}, and
\beq
\label{Q def}
Q(z) := \begin{cases}I, &\mbox{$\Im z> 0$}\\
\begin{pmatrix}
 0 & -1 \\
 1 & 0 \\
\end{pmatrix}, &\mbox{$\Im z < 0$.} \end{cases}
\eeq
\end{itemize}
In the above the principal branches of $z^{1/2}$, $z^{-1/4}$ and $\log z$ are taken analytic off $(-\infty,0]$ and positive for $z>0$.

\item[(d)] As $z \to b \in B$,
\beq
\Phi(z) = \bigO(1)\exp\left[-\sum^{\hat{d}_{b}-1}_{j=1}\frac{1}{2j} \tau^j_{b,j} (z-b)^{-j}\sigma_3\right][z-b]^{\hat{\alpha}_b \sigma_3}.
\eeq
As $z \to a \in \partial I \setminus B $
\beq
\Phi(z) = \Phi_{a,0}(z) \times \begin{cases} e^{\frac{1}{2\pi i } \log (z-a) \sigma_+}, & \mbox{if $a$ is a right edge of an interval,}\\e^{-\frac{1}{2\pi i }\log (a-z) \sigma_+} , & \mbox{if $a$ is a left edge of an interval.} \end{cases}
\eeq
Here $\Phi_{a,0}$ is analytic with respect to $z$ at $a$.
\end{itemize}
\edf

\begin{remark}
The above model problem has jump matrices which are constant with respect to all parameters appearing in the problem. This has the consequence that Lax matrices may be constructed from the logarithmic derivatives, $\partial \Phi \Phi^{-1}$, with respect to any parameter appearing in the problem. The compatibility of these derivatives then leads to a system of ODEs which will possess the Painlev\'{e} property. Hence the model problem is a Riemann-Hilbert problem for ODEs of Painlev\'{e} type.
\end{remark}

\subsubsection{Definitions for the double scaling limit}

\begin{definition}
Recall the definition of the model \eqref{semiclassical mu}. Let $\C{A} := \partial \C{I}$. This set contains the endpoints of the intervals comprising $\C{I}$.
\end{definition}

\begin{definition}
Given a singular point $\hat{z}$ of $\frac{dV}{dz}$ we write the principal
part of the Laurent expansion of $\frac{dV}{dz}$ at $\hat{z}$, whose inner radius is zero, as
\beq
\sum^{d_{\hat{z}}-1}_{j=0} t_{\hat{z},j} (z-\hat{z})^{-j-1}.
\eeq
Let $\C{B}$ be the set of all singular points of $\frac{dV}{dz}$. Since $V'(z)$ is a rational function we may parameterise it as,
\beq
V(z) := V_\mathrm{reg}(z) + V_\mathrm{sing}(z) + V_\mathrm{br}(z),
\eeq
where,
\begin{align}
V_\mathrm{reg}(z) &:= \sum^{d_\infty+1}_{j=1}\frac{1}{j} t_{\infty,j} z^j,\\
V_\mathrm{sing}(z) &:= -\sum_{b \in \C{B}} \left [\sum^{d_b-1}_{j=1}\frac{1}{j} t^j_{b,j} (z-b)^{-j} \right],\\
V_\mathrm{br}(z) &:= -\sum_{b \in \C{B}} \frac{2}{n}\alpha_b \log |z-b|.
\end{align}
Note that the fact that the probability measure for the model must be real implies $b \in \C{B} \iff b^* \in \C{B}$. For convenience we define some shorthand notation; $w(z):=e^{-n V(z)}$, $w_\mathrm{reg}(z):=e^{-n V_\mathrm{reg}(z)}$, $w_\mathrm{sing}(z):=e^{-n V_\mathrm{sing}(z)}$, $w_\mathrm{br}(z):=e^{-n V_\mathrm{br}(z)}$.
We restrict ourselves to the case that $\alpha_b$ is a positive constant. Finally we define $V_\mathrm{reg}$ such that,
\beq
\lim_{x\to +\infty}\frac{V_\mathrm{reg}(x)}{\log(x^2+1)}=+\infty,
\eeq
and that $V_\mathrm{sing}$ is such that the integral exists.
\end{definition}

\begin{definition} The measure $d\mu(x)$ is defined as the equilibrium measure which minimizes
\beq
I(\mu) = \iint \log \frac{1}{|x-y|} d\mu(x)d\mu(y) + \int V_\mathrm{reg}(y) d\mu(y),
\eeq
among all Borel probability measures $\mu$ on $\C{I}$. The equilibrium measure can be written in terms of
a density $\rho$; $d\mu(x)=\rho(x)dx$. Define $\C{S} := \supp \mu$.
\end{definition}

\begin{remark} The equilibrium measure satisfies the following inequalities,
\begin{align}
&\label{var eq}2\int \log|x-y|\rho(y)dy - V(x) = \ell, \quad x\in \C{S}, \\
&\label{var ineq}2\int \log|x-y|\rho(y)dy - V(x) \leq \ell, \quad x\in \C{I} \setminus \C{S}.
\end{align}
\end{remark}

\begin{definition}
The inequality \eqref{var ineq} is taken to be strict for $x \in \C{I} \setminus \C{S}$  outside a finite number of isolated points. We define $\C{E} \subset  \C{I} \setminus \C{S}$ to be the set of such isolated exterior points and $\C{J} := \C{S}\cup \C{E}$. Finally let $p:= \sup \C{J}$.
\end{definition}

\begin{remark}
\label{point types}
Since the equilibrium measure is defined using $V_\mathrm{reg}$ the behaviour of $\rho$ is well known \cite{DKM,DKMVZ2} and can be expressed as,
\beq
\label{rho h R}
\rho(x) = \frac{1}{2\pi i} y(x)_+,
\eeq
where we have introduced the {\em{spectral curve}},
\beq
y(z):= h(z)\sqrt{R(z)}.
\eeq
Here $h$ is a real polynomial and $R$ is a rational function with only simple poles. The set of zeros and poles of $R$ coincides with $\partial \C{S}$ and the function $\sqrt{R(z)}$ is defined with branch cuts coinciding with $\C{S}$. In \eqref{rho h R} we have taken the value of $\sqrt{R(z)}$ on the positive side of the cut. The behaviour of $\rho$ is best discussed in terms of three distinct types of points:
\begin{itemize}
\item An edge point $x$ of order $k\in \{-1\}\cup 2\mathbb{N}^0$ is a point such that $x \in \partial\C{S}$ and
\beq
y(z) = \bigO((z-x)^{k+\frac{1}{2}})  \qquad \mbox{ as $z \to x$.}
\eeq
A {\em{singular}} edge is one whose order $k > 0$, a {\em {hard-edge}} is one for which $k=-1$ and a {\em{soft-edge}} is one with $k=0$.
\item An interior point $x$ of order $k\geq 0$ is a point such that $x \in \C{S}$ and
\beq
y(z) = \bigO((z-x)^{2k})  \qquad \mbox{ as $z \to x$.}
\eeq
A singular interior point is one whose order $k > 0$.
\item An exterior point $x$ of order $k\geq 1$ is a point such that $x \in \C{E}$, i.e. a point at which the \eqref{var eq} holds outside of $\C{S}$, and
\beq
y(z) = \bigO((z-x)^{2k-1})  \qquad \mbox{ as $z \to x$.}
\eeq
Such points are also known as ``a birth of a cut''.
\end{itemize}
\end{remark}

\bdf
Let $\C{R}$ be the set of zeros and poles of $R$ and $\C{H}$ be the set of zeros of $h$.
\edf

\begin{definition}
Define the ``$g$-function'' \beq g(z) := \int
\log(z-x) d\mu(x), \eeq where the principal branch of the logarithm is
taken, meaning $g$ is analytic on $\mathbb{C} \setminus (-\infty, p]$.
The $g$-function has a number of properties we will make use of, in
particular
\begin{align}
&g_+(x) + g_-(x) - V(x)-\ell = 0 \qquad x\in \C{J} \label{g1},\\
&g_+(x) + g_-(x) - V(x)-\ell < 0 \qquad x\in \C{I} \setminus \C{J},\label{g2} \\
&g_+(x) - g_-(x) = 2\pi i \int^{p}_x \rho(x) dx.\label{g3}
\end{align}
Let us also define,
\beq
\label{xi} \xi(z) = -\frac{ 1}{2}
\int^z_{p}h(s) \sqrt{R(s)} ds \qquad \text{for $z  \in \mathbb{C} \setminus (-\infty, p]$,}
\eeq
where the integration path does not cross the real
axis.  By (\ref{g3}) we have, 
\beq\label{g xi 1}
g_+(x) - g_-(x) = 2\xi_+(x),
\eeq
and using (\ref{g1}), we obtain the identity
\beq
\label{g xi 2} 2\xi(z)=2g(z)-V(z)-\ell. 
\eeq
The jumps of $\xi$ follow from those of $g$,
\begin{align}
&\xi_+(x) + \xi_-(x) = 0 \qquad x\in \C{J} \label{xi1},\\
&\xi_+(x) + \xi_-(x) < 0 \qquad x\in \C{I} \setminus \C{J},\label{xi2} \\
&\xi_+(x) - \xi_-(x) = 2\pi i \int^{p}_x \rho(x) dx.\label{xi3}
\end{align}
It is important to note that \eqref{xi3} has the consequence that,
\beq
\xi_+(x) - \xi_-(x) = 2 \pi i \epsilon(x), \qquad x \in \mathbb{R} \setminus \C{J}
\eeq
where $\epsilon:\mathbb{R} \rightarrow [0,1]$ is constant on each connected component of $\mathbb{R} \setminus \C{J}$ and is known as the ``filling fraction''.

\end{definition}

\begin{assumptions}
\label{ass1}
The sets $\C{A}$, $\C{B}$ etc. are $n$ dependent. For an any finite set $X$ of $n$ dependent points we assume that the limit,
\beq
X_* := \lim_{n\rightarrow \infty} X
\eeq
exists. For some $x \in X$ we also use the notation $x_* := \lim_{n\rightarrow \infty} x$. We extend these definitions to intervals and collections of intervals by defining $[a, b]_* := [a_*, b_*]$. We also assume that $\C{A}_* \cup \C{B}_* \subset \C{J}_*$ i.e all singular points that are due to the weight collide with the support of the eigenvalues as $n \rightarrow \infty$. 
\end{assumptions}

The singular behaviour of the weight in \eqref{semiclassical mu} is not the only way in which singular points of the spectrum arise. Indeed the most common way is when a zero of $h$ approaches $\C{J}$. We therefore need to keep track of these zeros of $h$ which collide with $\C{J}$ as $n \rightarrow \infty$. We do this via the following definition.

\begin{definition}
\label{hatH}
Let us define $\hat{\C{H}}$ to be the largest subset of $\C{H}$ such that $\hat{\C{H}_*} \subset  \C{J}_*$. It will be useful to collect all the points that cause singular behaviour, besides zeros of $\C{R}$, together in a single set,
\beq
\C{P} := \C{A} \cup \C{B} \cup \hat{\C{H}}.
\eeq
\end{definition}

\begin{definition}
\label{Delta def}
For $x \in \C{P} \cup \C{R}$ let $x_* :=  \lim_{n \rightarrow \infty} x$. Given Assumption \ref{ass1} we have $x_* \in \C{J}_*$ and therefore that $x_*$ can be classified according to the scheme in Remark \ref{point types}. We say $x$ {\em{scales appropriately}} as $n \rightarrow \infty$ if $|x - x_*| = \bigO(n^{-\Delta_{x_*}})$ where $\Delta_{x_*}$ is defined by,
\beq
\Delta_{x_*} := \begin{cases}
\frac{2}{2k+3},\; & \mbox{if $x_*$ is a edge point of order $k$,}\\
\frac{1}{2k+1},\; & \mbox{if $x_*$ is an interior point of order $k$,}\\
\frac{1}{2k},\; & \mbox{if $x_*$ is an exterior point of order $k$.}
\end{cases}
\eeq
Furthermore, given $b \in \C{B}$ we say the parameters $t_b$ scale appropriately if $ \lim_{n \rightarrow \infty} n^{\Delta_{x_*}} t_b$ exists. 
\end{definition}

\begin{assumptions}
\label{assumption: scale appropriately}
We assume all points in $\C{P}\cup\C{R}$ scale appropriately.
\end{assumptions}

\bdf
\label{bar x_star def}
Given a set $X \subset \C{P}\cup \C{R}$ and $x_* \in \C{P}_* \cup \C{R}_*$ we define $X|_{x_*} := \{z \in X : z_*  = x_*\}$. For a collection of intervals $I$ we define $I|_{x_*}$ to be the collection of intervals such that $\partial(I|_{x_*}) = (\partial I)|_{x_*}$. As it stands, the definition of $I|_{x_*}$ is ambiguous as it only specifies the endpoints of the intervals. We take care of this by further requiring that if $x_*$ is an exterior point, then $I|_{x_*}$ does not have intervals extending to $\pm\infty$, if $x_*$ is an interior point then $I|_{x_*}$ has intervals extending to both $\pm \infty$ and if $x_*$ is an edge point then $I|_{x_*}$ extends to only $-\infty$.
\edf

\bdf
\label{f def}
For $x_* \in \C{P}_* \cup \C{R}_*$, define $f(z) := n^{\Delta_{x_*}}(z - x_*)$. The definition of $f$ extends to intervals by defining $f([x_0,x_1]) := [f(x_0),f(x_1)]$ with similar definitions for intervals with open ends.
\edf

\bdf
\label{scaling y def}
We define the scaling limit of $y$ at $x_*$ by,
\beq
\label{scaling y}
\hat{y}(\zeta) := \lim_{n \rightarrow \infty} n^{1-\Delta_{x_*}} y(x_* + n^{-\Delta_{x_*}} \zeta)\times \begin{cases} 1, & \mbox{$x_*$ is a exterior or edge point,} \\ \theta(\zeta),& \mbox{$x_*$ is an interior point.}  \end{cases}
\eeq
Define also,
\beq
\hat{\xi}(\zeta) := -\frac{1}{2}\int^\zeta_{\hat{p}} \hat{y}(z) dz,
\eeq
where $\hat{p} := \lim_{n\rightarrow \infty}f(p')$ and $p' := \sup \C{R}|_{x_*}$ if $\sup \C{R}|_{x_*}$ exists else $p' = x_*$.
\edf

\bdf
\label{model problem at xs}
We define {\em{the model problem at}} $x_*$, $\Phi^{(x_*)}(z)$, to be the model problem of the same type and order as the point $x_*$. The vector $\vec{\tau}_{\infty}$ is determined using the scaling limit of $y$ at $x_*$. We choose $\vec{\tau}_\infty$ in the following way,
\begin{itemize}
\item[(i)] If $x_*$ is an interior point of order $k$ we have for large $\zeta$,
\beq
\hat{\xi}(\zeta)  = -\frac{i}{2k+1} \sum^{2k+1}_{j = 0} E_j\zeta^j + \bigO(\zeta^{-1}),
\eeq
where $E_j$ are constants. We set $\tau_{\infty,0} = E_{0} + (2k+1) ( n i \xi_+(x_*) + \pi  \alpha^+_{x_*} )$ and for $j>0$, $\tau_{\infty,j} = E_j$.
\item[(ii)] If $x_*$ is an edge point of order $k$ we have for large $\zeta$,
\beq
\zeta^{-\frac{1}{2}}\hat{\xi}(\zeta)  =  -\frac{2}{2k+3} \sum^{k+1}_{j = 0} E_j \zeta^j +\bigO(\zeta^{-1}),
\eeq
where $E_j$ are constants. We set $\tau_{\infty,j} = E_j$.
\item[(iii)] If $x_*$ is an exterior point of order $k$ we have for large $\zeta$,
\beq
\hat{\xi}(\zeta) = - \hat{c} \log\zeta - \frac{1}{2k}\sum^{2k}_{j = 0}E_j \zeta^j + \bigO(\zeta^{-1}),
\eeq
where $E_j$ are constants. We set $\tau_{\infty,j} = E_j$.
\end{itemize}

The remaining data for the model problem is,
\begin{align}
I = & \lim_{n\rightarrow\infty} f(\C{I}|_{x_*}),\\
B = & \lim_{n\rightarrow\infty} f(\C{B}|_{x_*}), \\
\tau_{b} = & \lim_{n\rightarrow\infty}  n^{\Delta_{x_*}} t_{f^{-1}(b)}.
\end{align}
\edf

\subsubsection{Main theorems}
\begin{theorem}
\label{existence thm}
A solution of the canonical model problem for $\Phi$ exists in the {\em{edge}} and {\em{interior}} case for any real vectors $\tau_b$ subject to the constraint that if $b \in B \cap \mathbb{R}$ there exists a constant $\delta \in \mathbb{R}$ such that,
\beq
\label{singbound}
\left(\frac{\tau_{b,\hat{d}_b}}{ z-b} \right)^{\hat{d}_b} > \delta, \qquad \forall z \in I.
\eeq

When $\Phi$ is of exterior type of order $k$ a solution exists if the above constraint is satisfied and
\beq
\int_I e^{-\frac{1}{2k} P_{k+1}(x)} dx
\eeq
exists.

\end{theorem}

\begin{remark} In \eqref{singbound} recall that $\hat{d}_b$ is the order of the pole in $V'$ appearing at $b$. The condition \eqref{singbound} ensures that if an essential singularity of the weight is approached along a contour in $I$, $\Phi$ remains bounded. The extra condition on $\Phi$ in the exterior case is related to the fact that in this case $\Phi$ may be constructed explicitly from orthogonal polynomials.
\end{remark}

The second theorem concerns the behaviour of the kernel near a point $x_* \in  \C{P}_* \cup \C{R}_*$. 

\begin{theorem}\label{Kthm} Consider a semi-classical matrix model with $\alpha_b >0$ for all $b \in B$. Given a point $x_* \in \C{P}_*\cup \C{R}_*$ define the analytic functions $\phi_i :\mathbb{H} \setminus I \setminus \cup_{b \in B} \Gamma_{b}:  \rightarrow \mathbb{C}$ by,
\beq
\label{phi def}
\begin{pmatrix} \phi_1(z) \\ \phi_2(z)\end{pmatrix} := \Phi^{(x_*)}(z) \begin{pmatrix}1 \\ 0 \end{pmatrix},
\eeq
where we have used the model problem at $x_*$ defined in Definition \ref{model problem at xs}. We also define the $\Phi$-kernel,
\begin{equation}
\mathbb K^{\Phi}(u,v) := -e^{-\pi i(\hat{\alpha}_\Gamma(u) + \hat{\alpha}_\Gamma(v))}\frac{\phi_1(u)\phi_2(v) - \phi_1(v)\phi_2(u)}{2 \pi i (u-v)}.
\end{equation}
In the double scaling limit where $n\to\infty$ such that all parameters scale appropriately we have that,
\beq
\lim_{n\to\infty}n^{-\Delta_{x_*}}K_n\left(x_* + n^{-\Delta_{x_*}} u,x_* + n^{-\Delta_{x_*}} v\right) = \mathbb K^{\Phi}(u,v),
\eeq
for $u,v \in  \lim_{n \rightarrow \infty} f(\C{I}|_{x_*})$.
\end{theorem}

\section{Properties of the canonical model problem $\Phi$}

\subsection{Opening the lens}
In order to compare to known model problems it is useful to transform the model
problem by opening the lens. To this end we introduce some lens-like contours.

\begin{figure}[t]
\centering 
\includegraphics[scale=0.4]{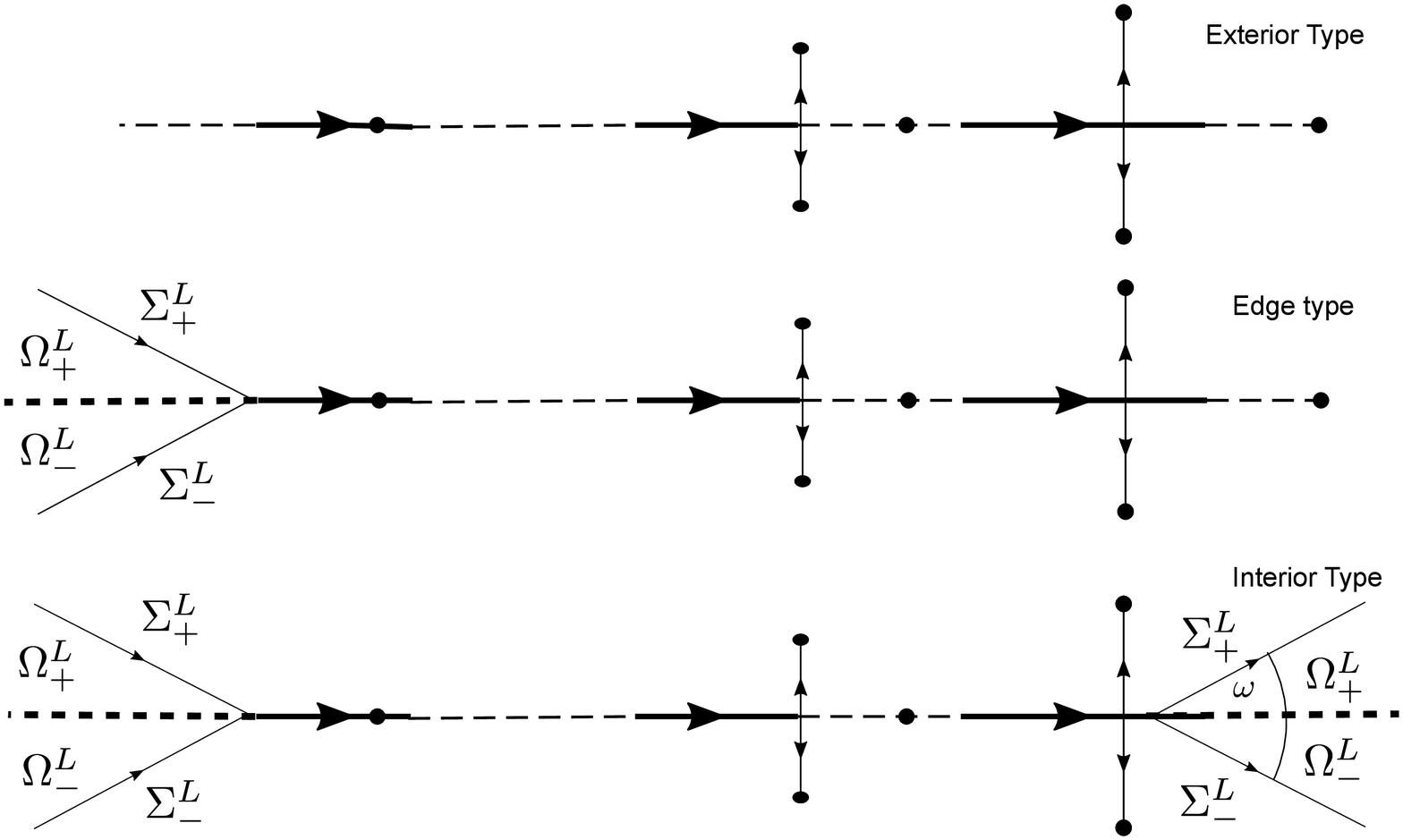}
\caption{Examples of jump contours for the canonical model problem with open lenses. The thick lines are the intervals in $I$. The thin dashed lines are the intervals in $I \cup \bigcup_{b \in {B}} \Gamma_b \setminus L$. The thick dashed line are the intervals in $L$. We have also labelled the lens contours $\Sigma^L_\pm$ and the lens regions $\Omega^L_\pm$. The black dots are locations of points in $B$ and the contours ending at these points are the $\Gamma_b$ contours.}
\label{modelcontours}
\end{figure}

\bdf
If $\Phi$ is of interior type let $L := (-\infty,p_1) \cup (p_2,\infty) \subset I$ for $p_1 \leq \inf \Re B \leq \sup \Re B \leq p_2$. If $\Phi$ is of edge type let $L := (-\infty,p_1) \subset I$ with $p_1 \leq \inf \Re B$.
\edf
The set $L$ will define the parts of the real line enclosed in the lens. Note that in the above definition we required that the lens not include any points from $B$.

\bdf
If $\Phi$ is of interior type let $\Sigma^{(-\infty,p_1)}_\pm$ be rays inside $\Omega^{(\mathrm{int},-)}$ from $p_1$ to $-\infty$ such that $\pm \Im z >0$ on each ray respectively and orientated away from infinity. Let $\Sigma^{(p_2,\infty)}_\pm$ be rays inside $\Omega^{(\mathrm{int},+)}$ from $p_2$ to $\infty$ such that $\pm \Im z >0$ on each ray respectively and orientated towards infinity. Finally let $\Sigma^L_\pm := \Sigma^{(-\infty,p_1)}_\pm \cup \Sigma^{(p_2,\infty)}_\pm$.
\edf

\bdf
If $\Phi$ is of edge type let $\Sigma^L_\pm$ be rays inside $\Omega^{(\mathrm{edge})}$ from $p_1$ to $-\infty$ such that $\pm \Im z >0$ on each ray respectively and with orientation away from infinity.
\edf

\bdf
We define $\Omega^{L}_\pm$ to be the region of $\mathbb{C}$ bounded by $L \cup \Sigma^L_\pm$ such that the $+$ side and $-$ side of $L$ is contained in $\Omega^{L}_+$ and $\Omega^{L}_-$ respectively.
\edf

We now make the transformation,
\beq
\widehat{\Phi}_k(z|I,B,\{\vec{\tau}_b\},\vec{\tau}_\infty) := \Phi_k(z|I,B,\{\vec{\tau}_b\},\vec{\tau}_\infty) \hat{K}(z, \Omega),
\eeq
where $\Omega = \Omega^L_- \cup \Omega^L_+$. This leads to the RH problem,
\subsubsection*{RH problem for $\widehat{\Phi}$}
\begin{itemize}
\item[(a)] $\widehat{\Phi}:\mathbb{C} \setminus I \setminus \Sigma^L_+  \setminus  \Sigma^L_- \setminus  \cup_{b \in B} \Gamma_{b} \rightarrow \mathbb{C}^{2\times2}$ is analytic in $z$ (see Figure \ref{modelcontours}).

\item[(b)] The jump matrix $j(z) := \widehat{\Phi}_-(z)^{-1}\widehat{\Phi}_+(z) $ has the following form,
\begin{align}
j(z) &= \left(
\begin{array}{cc}
 0 & 1 \\
 -1 & 0 \\
\end{array}
\right), &z\in L,\\
j(z) &= \left(
\begin{array}{cc}
 e^{2 i \pi  \hat{\alpha}_\Gamma(z)} & \chi_I(z) \\
 0 & e^{-2 i \pi  \hat{\alpha}_\Gamma(z)} \\
\end{array}
\right), &z\in I \cup \bigcup_{b \in {B}} \Gamma_b \setminus L,\\
j(z) &= \left(
\begin{array}{cc}
 1 & 0 \\
 e^{\pm2 i \pi  \hat{\alpha}_\Gamma(z)} & 1 \\
\end{array}
\right), &z \in \Sigma^L_\pm.
\end{align}
\item[(c)] As $z\rightarrow \infty$,
\begin{itemize}
\item[(i)] If $\widehat{\Phi}$ is of exterior-type,
\beq
\label{Phi z to inf ex lens}
\widehat{\Phi}(z)=\left(I+\bigO(z^{-1})\right) z^{(\hat{\alpha}_\mathrm{tot} + \hat{c} ) \sigma_3} e^{-\frac{1}{2k} \hat{P}_{2k}(z) \sigma_3}.
\eeq

\item[(ii)] If $\widehat{\Phi}$ is of edge-type,
\beq
\label{Phi z to inf edge lens}
\widehat{\Phi}(z)=\left(I+\bigO(z^{-1})\right) z^{-\frac{\sigma_3}{4}} N e^{-\frac{2}{2k+3} z^\frac{1}{2}\hat{P}_{k+1}(z) \sigma_3}.
\eeq

\item[(iii)] If $\widehat{\Phi}$ is of interior-type,
\beq
\label{Phi z to inf int lens}
\widehat{\Phi}(z)=\left(I+\bigO(z^{-1})\right) e^{-\frac{i}{2k+1}\hat{P}_{2k+1}(z) \sigma_3} Q(z).
\eeq
Note that we have made use of the fact that $Q(z)e^{a \theta(z) \sigma_3} = e^{a\sigma_3} Q(z)$. 
\end{itemize}

\item[(d)] As $z \to b \in B$,
\beq
\widehat{\Phi}(z) = \bigO(1)\exp\left[-\sum^{\hat{d}_{b}-1}_{j=1}\frac{1}{2j} \tau^j_{b,j} (z-b)^{-j}\sigma_3\right][z-b]^{\hat{\alpha}_b  \sigma_3} \hat{K}(z,\Omega).
\eeq
As $z \to a \in \partial I \setminus B $
\beq
\widehat{\Phi}(z) = \widehat{\Phi}_{a,0}(z) \times
\begin{cases} e^{\frac{1}{2\pi i } \log (z-a) \sigma_+}\hat{K}(z,\Omega), & \mbox{if $a$ is a right edge of an interval,}\\
e^{-\frac{1}{2\pi i }\log (a-z) \sigma_+} \hat{K}(z,\Omega), & \mbox{if $a$ is a left edge of an interval.} \end{cases}
\eeq
\end{itemize}

\begin{remark}
The $\widehat{\Phi}$ RH problem follows straightforwardly from the definition of $\Phi$. However there is one slightly subtle point when computing the asymptotic behaviour in property (c). For $z$ in $\Omega^L_\pm$ and $z \in \mathbb{C} \setminus \Omega$, where $\Omega$ is taken to be $\Omega^{(\mathrm{edge})}$ or $\Omega^{(\mathrm{int})}$ for edge and interior type problems respectively, 
the asymptotics claimed in (c) follows immediately. However if $z \rightarrow \infty$ for $z \in \Omega \setminus \Omega^L_+ \setminus \Omega^L_-$ one still has the asymptotic behaviour in \eqref{Phi z to inf edge} or \eqref{Phi z to inf int}, with the associated $\hat{K}$ factor. However, the $\hat{K}$ factor can be removed by conjugating $\hat{K}$ through to the left and noting it contributes an error term smaller than $\bigO(z^{-1})$.
\end{remark}

\subsection{Examples}
In this section we will give some examples showing how the canonical model problem reduces to familiar model RH problems used previously in the random matrix theory literature. We make use of the open lens version of the canonical model problem introduced above.

\subsubsection{Relation to the Airy model problem}
Setting $I = \mathbb{R}$, $L = \mathbb{R}^-$, $B = \varnothing$ and $\tau_{\infty,0} = 1$ in a edge-type problem of order $k =0$ we have the following RH problem,
\begin{itemize}
\item[(a)] $\widehat{\Phi} : \mathbb{C} \setminus \mathbb{R} \setminus \Sigma^{\mathbb{R}^-}_+\setminus \Sigma^{\mathbb{R}^-}_- \rightarrow \mathbb{C}^{2 \times 2}$ is analytic in $z$.
\item[(b)] $\widehat{\Phi}$ has the jumps shown in Figure \ref{airyjumps}
\item[(c)] As $z \rightarrow \infty$,
\beq
\widehat{\Phi}(z)=\left(I+\bigO(z^{-1})\right) z^{-\frac{\sigma_3}{4}} N e^{-\frac{2}{3} z^\frac{3}{2}\sigma_3}.
\eeq
\end{itemize}
This is exactly the Airy model problem \cite{DKMVZ2}.

\begin{figure}[t]
\centering 
\includegraphics[scale=0.4]{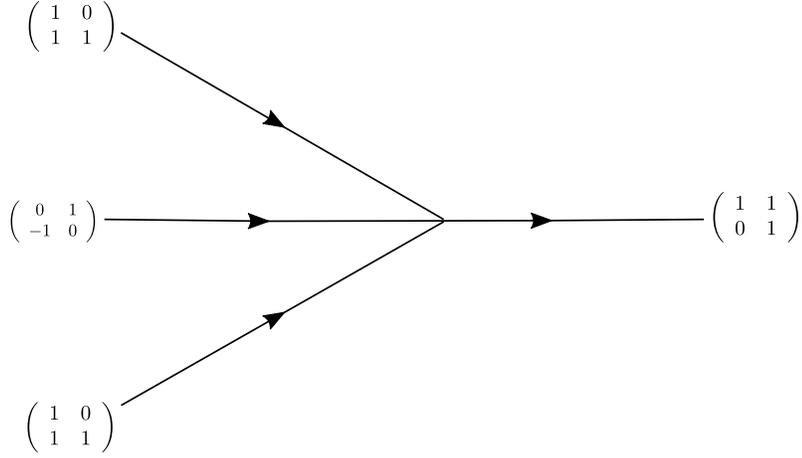}
\caption{Jumps for the canonical model problem in the Airy case.}
\label{airyjumps}
\end{figure}

\begin{figure}[t]
\centering 
\includegraphics[scale=0.4]{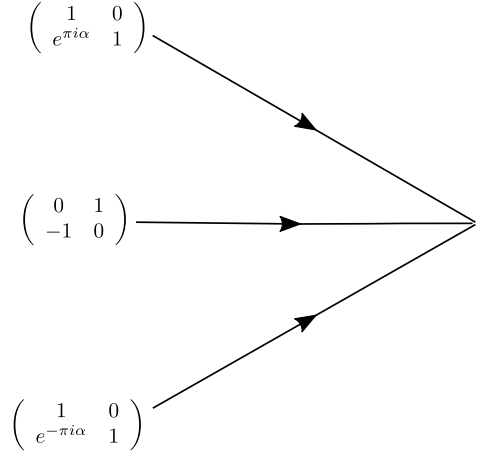}
\caption{Jumps for the canonical model problem in the Bessel case at a hard edge with a logarithmic singularity at the edge.}
\label{besseljumps1}
\end{figure}

\subsubsection{Relation to the Bessel model problem}
Setting $I = \mathbb{R}^{-}$, $L = \mathbb{R}^-$, $B = \{0\}$, $\hat{\alpha}_0 = \alpha/2$ and $\tau_{\infty,0} = -1$ in a edge-type problem of order $k=-1$, we have the following RH problem,
\begin{itemize}
\item[(a)] $\widehat{\Phi} : \mathbb{C} \setminus \mathbb{R}^- \setminus \Sigma^{\mathbb{R}^-}_+\setminus \Sigma^{\mathbb{R}^-}_- \rightarrow \mathbb{C}^{2 \times 2}$ is analytic in $z$.
\item[(b)] $\widehat{\Phi}$ has the jumps shown in Figure \ref{besseljumps1}
\item[(c)] As $z \rightarrow \infty$,
\beq
\widehat{\Phi}(z)=\left(I+\bigO(z^{-1})\right) z^{-\frac{\sigma_3}{4}} N e^{2 z^\frac{1}{2}\sigma_3}.
\eeq
\item[(d)] As $z \rightarrow 0$,
\beq
\widehat{\Phi}(z) = \bigO(1)(z-b)^{\frac{\alpha}{2}  \sigma_3} \hat{K}(z).
\eeq
This can be written as,
\beq
\widehat{\Phi}(z)=\bigO \begin{pmatrix}z^\frac{\alpha}{2} & z^{-\frac{\alpha}{2}}\\  z^\frac{\alpha}{2} & z^{-\frac{\alpha}{2}} \end{pmatrix},
\eeq
for $z$ outside the lens and
\beq
\widehat{\Phi}(z)=\bigO \begin{pmatrix}z^{-\frac{\alpha}{2}} & z^{-\frac{\alpha}{2}}\\  z^{-\frac{\alpha}{2}} & z^{-\frac{\alpha}{2}} \end{pmatrix}.
\eeq
for $z$ inside the lens.
\end{itemize}
This is exactly the Bessel model problem.

\begin{figure}[t]
\centering 
\includegraphics[scale=0.4]{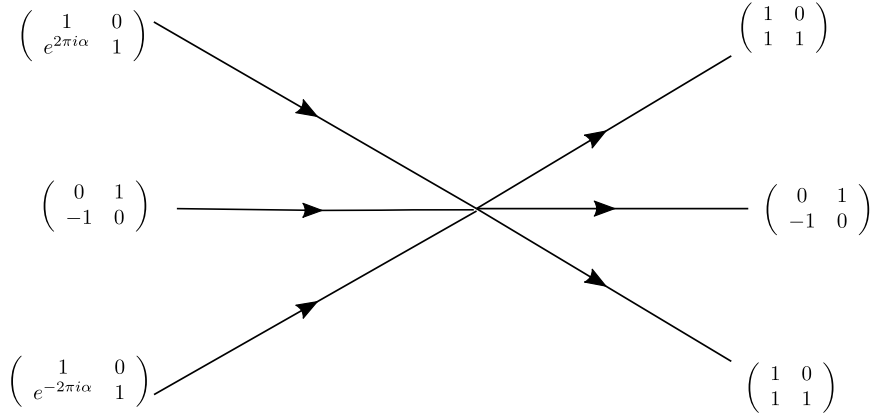}
\caption{The jumps for the canonical model problem in case of Bessel or Painlev\'{e} II in the bulk with a logarithmic singularity.}
\label{besseljumps2}
\end{figure}

\subsubsection{Relation to Bessel in the bulk}
Setting $I = \mathbb{R}$, $L = \mathbb{R}\setminus\{0\}$, $B = \{0\}$, $\hat{\alpha}_0 = \alpha$, $\tau_{\infty,1} = 1$ and $\tau_{\infty,0} = -\frac{1}{2}\pi \alpha$ in an interior-type problem of order $k=0$ we have the following RH problem,
\begin{itemize}
\item[(a)] $\widehat{\Phi} : \mathbb{C} \setminus \mathbb{R} \setminus \Sigma^{\mathbb{R}\setminus\{0\}}_+\setminus \Sigma^{\mathbb{R}\setminus\{0\}}_- \rightarrow \mathbb{C}^{2 \times 2}$ is analytic in $z$.
\item[(b)] $\widehat{\Phi}$ has the jumps shown in Figure \ref{besseljumps2}
\item[(c)] As $z \rightarrow \infty$,
\beq
\widehat{\Phi}(z)=\left(I+\bigO(z^{-1})\right) e^{-i(z +  \frac{1}{2}i\pi \alpha)\sigma_3} Q(z).
\eeq
\end{itemize}
The above RH problem is related to that in \cite{Vanlassen0305044} by,
\beq
\Psi(z) = \widehat{\Phi}(z)\times \begin{cases} e^{-\pi i \alpha \theta(z) \sigma_3}, &\mbox{if $\Re z > 0$,} \\ 1, &\mbox{if $\Re z < 0$.} \end{cases}
\eeq

\subsubsection{Relation to the general Painlev\'{e} II equation}
Setting $I = \mathbb{R}$, $L = \mathbb{R}\setminus\{0\}$, $B = \{0\}$, $\hat{\alpha}_0 = \alpha$, $\tau_{\infty,3} = 4$, $\tau_{\infty,2} = 0$ and $\tau_{\infty,1} = 3s$ in an interior-type problem of order $k=1$ we have the following RH problem,
\begin{itemize}
\item[(a)] $\widehat{\Phi} : \mathbb{C} \setminus \mathbb{R} \setminus \Sigma^{\mathbb{R}\setminus\{0\}}_+\setminus \Sigma^{\mathbb{R}\setminus\{0\}}_- \rightarrow \mathbb{C}^{2 \times 2}$ is analytic in $z$.
\item[(b)] $\widehat{\Phi}$ has the jumps shown in Figure \ref{besseljumps2}
\item[(c)] As $z \rightarrow \infty$,
\beq
\widehat{\Phi}(z)=\left(I+\bigO(z^{-1})\right) e^{-\frac{i}{3}(4z^3 +  3s z + \tau_{\infty,0} )\sigma_3} Q(z).
\eeq
\end{itemize}
The above RH problem is related to that in \cite{claeys2008multi} by,
\beq
\label{PIIrel}
\Psi^{PII}(z) = e^\frac{\pi i \alpha \sigma_3}{2} e^\frac{i \tau_{\infty,0}}{3}\widehat{\Phi}(z)Q(z)^{-1} e^{-\frac{\pi i \alpha \sigma_3}{2}}.
\eeq
\subsubsection{Relation to GUE at a birth of a cut}
Setting $I = \mathbb{R}$, $B = \varnothing$, $\tau_{\infty,2} = 1$, $\tau_{\infty,1} = 0$ and $\tau_{\infty,0} = 0$ in an exterior-type problem of order $k=1$ we have the following RH problem,
\begin{itemize}
\item[(a)] $\widehat{\Phi} : \mathbb{C} \setminus \mathbb{R} \rightarrow \mathbb{C}^{2 \times 2}$ is analytic in $z$.
\item[(b)] $\widehat{\Phi}$ has the jump,
\beq
\widehat{\Phi}^{-1}(z) \widehat{\Phi}(z) = \begin{pmatrix}1 & 1\\ 0&1 \end{pmatrix}
\eeq
for $z \in \mathbb{R}$.
\item[(c)] As $z \rightarrow \infty$,
\beq
\widehat{\Phi}(z)=\left(I+\bigO(z^{-1})\right) z^{\hat{c} \sigma_3} e^{-\frac{z^2}{2} \sigma_3}.
\eeq
This RH problem is related to the standard RH problem for Hermite polynomials of order $\hat{c}$ by,
\beq
\widehat{\Phi}(z) = Y(z) e^{-\frac{z^2}{2} \sigma_3}.
\eeq
\end{itemize}

\subsection{Solvability of $\Phi$ RH problem: Proof of theorem \ref{existence thm}}
We follow the standard argument used in \cite{IKO}. In particular the existence of a solution to the model Riemann Hilbert problem depends crucially on showing a ``vanishing lemma''. To state it, we define a function $\Lambda$ so that we may write the asymptotics \eqref{Phi z to inf edge} and \eqref{Phi z to inf int} in a uniform manner. In particular let,
\beq
\Lambda(z) := \begin{cases} -\frac{2}{2k+3} z^\frac{1}{2}\hat{P}_{k+1}(z), & \mbox{if $\Phi$ is of edge type,}\\ -\frac{i}{2k+1}\hat{P}_{2k+1}(z) \theta(z)  & \mbox{if $\Phi$ is of interior type,} \end{cases}
\eeq
we then have the following result.

\begin{lemma}[Vanishing Lemma]
Let $\Phi_0$ be a function satisfying the conditions (a), (b) and (d) of the RH problem for $\Phi$ together
with the asymptotic behaviour as $z \rightarrow \infty$ of,
\beq
\Phi_0(z)=\bigO(z^{-\eta}) e^{\Lambda(z) \sigma_3} \hat{K}(z, \Omega)^{-1}
\eeq
where $\Omega$ takes the form specified in the RH problem for $\Phi$ in the edge and interior cases, $\eta = 3/4$ if $\Phi$ is of edge type and $\eta=1$ if $\Phi$ is of interior type. Then $\Phi_0 \equiv 0 $.
\end{lemma}
\begin{proof}
Let us first draw the readers attention to the following properties of $\Lambda$. Let $\hat{\Lambda}(x)$ be a function which is real for $x \in \mathbb{R}$. When $\Phi$ is edge type we have, 
\begin{align}
\Lambda_+(x) = \Lambda_-(x) &= -\hat{\Lambda}(x), & x > 0, \\
\Lambda_+(x) = -\Lambda_-(x) &= -i \hat{\Lambda}(x), & x < 0,
\end{align}
and when $\Phi$ is of interior type we have 
\beq
\Lambda_+(x) = -\Lambda_-(x) = - i \hat{\Lambda}(x), \qquad x \in \mathbb{R}.
\eeq
In the above we have used the fact that the vector $\tau_{\infty}$ is real. 

We now let,
\beq
G(z) := \Phi_0(z) e^{-\Lambda(z) \sigma_3} \hat{Q}(z)
\eeq
where we have used,
\beq
\hat{Q}(z) := \begin{cases} \begin{pmatrix}0 & -1 \\ 1& 0 \end{pmatrix}, & \Im z > 0,\\ I & \Im z < 0. \end{cases}.
\eeq
The function $G$ satisfies the RH problem,

\subsubsection*{RH problem for $G$}
\begin{itemize}
\item[(a)] $G:\mathbb{C} \setminus \mathbb{R}\setminus \cup_{b \in B} \Gamma_{b} \rightarrow \mathbb{C}^{2\times2}$ is analytic. 
\item[(b)]  The jump matrix $j_G(z) := G_-(z)^{-1}G_+(z) $ has the following form,
\begin{align}
j_G &= \left(
\begin{array}{cc}
\chi_{I}(z) e^{\Lambda_+(z)+\Lambda_-(z)} & -e^{\Lambda_-(z)-\Lambda_+(z)+2 \pi i  \hat{\alpha}_\Gamma(z)} \\
 e^{-\Lambda_-(z)+\Lambda_+(z)-2 \pi i  \hat{\alpha}_\Gamma(z)} & 0 \\
\end{array}
\right),\quad \mbox{if $z \in \mathbb{R}$}\\
j_G &= e^{\mp 2 \pi i \hat{\alpha}_{\Gamma}(z) \sigma_3},\quad \mbox{if $\pm\Im z > 0$ and  $z \in \bigcup_{b \in B} \Gamma_b$}
\end{align}
These jumps specialise in the edge and interior cases to:
\begin{itemize}
\item[(i)] $\Phi$ is of edge type:
\begin{align}
j_G&= \left(
\begin{array}{cc}
 e^{-2 \hat{\Lambda}(z)} \chi_I(z) & -e^{2 \pi i  \hat{\alpha}_\Gamma(z)} \\
 e^{-2 \pi i  \hat{\alpha}_\Gamma(z)} & 0 \\
\end{array}
\right), \quad z \in \mathbb{R}^+&\\
j_G&=\left(
\begin{array}{cc}
 \chi_I(z) & -e^{2 i \hat{\Lambda}(z)+2 \pi i  \hat{\alpha}_\Gamma(z)} \\
 e^{-2 i \hat{\Lambda}(z)-2 \pi i  \hat{\alpha}_\Gamma(z)} & 0 \\
\end{array}
\right),  \quad z \in \mathbb{R}^- &\\
j_G &= e^{\mp 2\pi i \hat{\alpha}_{\Gamma}(z) \sigma_3},\quad \mbox{if $\pm\Im z > 0$ and  $z \in \bigcup_{b \in B} \Gamma_b$}.&
\end{align}
\item[(ii)] $\Phi$ is of interior type:
\begin{align}
j_G&=\left(
\begin{array}{cc}
 \chi_I(z) & -e^{2 i \hat{\Lambda}(z)+2 \pi i  \hat{\alpha}_\Gamma(z)} \\
 e^{-2 i \hat{\Lambda}(z)-2 \pi i  \hat{\alpha}_\Gamma(z)} & 0 \\
\end{array}
\right), \quad z \in \mathbb{R}\\
j_G &= e^{\mp 2\pi i \hat{\alpha}_{\Gamma}(z) \sigma_3},\quad \mbox{if $\pm\Im z > 0$ and  $z \in \bigcup_{b \in B} \Gamma_b$}.&
\end{align}
\end{itemize}

\item[(c)] As $z\rightarrow \infty$,
\beq
G(z)=\bigO(z^{-\eta}),
\eeq
where $\eta = 3/4$ if $\Phi$ is of edge type and $\eta=1$ if $\Phi$ is of interior type. Note that this behaviour depends crucially on the constraint $\tau_{\infty,2k+1}>0$ in the interior case, $\tau_{\infty,k+1}>0$ in the edge case with $k \in 2 \mathbb{N}^0$ and $\tau_{\infty,k+1} < 0$ in the edge case with $k =-1$. 

\item[(d)] As $z \to b \in B$,
\beq
G(z) = \bigO(1) \exp\left[-\sum^{\hat{d}_{b}-1}_{j=1}\frac{1}{2j} \tau_{b,j}^j (z-b)^{-j}\sigma_3\right] [z-b]^{\hat{\alpha}_b \sigma_3}\hat{Q}(z)
\eeq
As $z \to a \in \partial I \setminus B$ for $\Im z< 0$ we have,
\beq
G(z) = \ \bigO\begin{pmatrix}1 & \log|z-a| \\ 1 & \log|z-a|\end{pmatrix},
\eeq
while for $z \to a \in \partial I \setminus B$ for $\Im z > 0$,
\beq
G(z) = \ \bigO\begin{pmatrix}\log|z-a| & 1 \\ \log|z-a| & 1\end{pmatrix}.
\eeq
\end{itemize}

We now introduce the function,
\beq
H(z) := G(z) G(z^*)^\dagger, \qquad z \in \mathbb{C} \setminus \mathbb{R}.
\eeq
Using the jumps of $G$ it is indeed possible to demonstrate that $H$ has no jump on $\cup_{b \in B} \Gamma_b \setminus \mathbb{R}$ and no singularities at $b \in B$ and is therefore analytic in $\mathbb{C} \setminus \mathbb{R}$. The only points at which some elements of $H$ are not bounded are in $\partial I \setminus B$, at which there may be $\log$ divergences. Lastly, we have that as $z\rightarrow \infty$,
\beq
H(z) = \bigO(z^{-2\eta}).
\eeq
Using the above properties together with Cauchy's theorem allows us to conclude,
\beq
\int_{\mathbb{R}} H_+(x)dx  = 0.
\eeq
By adding the hermitian conjugate of the above equation to itself we have,
\beq
\int_{I\cap(-\infty,0]} G_-(x)\begin{pmatrix}1& 0\\0 &0  \end{pmatrix} G_-(x)^\dagger dx + \int_{I\cap[0,\infty)} G_-(x)\begin{pmatrix}\C{V}(x)& 0\\0 &0  \end{pmatrix} G_-(x)^\dagger dx = 0,
\eeq
where $\C{V}(x) := e^{-2\hat{\Lambda}(x)}$ in the edge case and $\C{V}(x) := 1$ in the interior case. In both cases we are able to conclude that the first column and second column of $G(z)$ is zero in the lower and upper half plane respectively.

To analyse the other entries of $G$ we use the standard argument based on Carlson's theorem \cite{IKO}. We define the scalar functions,
\beq
g_k(z) := \begin{cases}G_{k,1}(z)& \Im z> 0\\ G_{k,2}(z)& \Im z< 0 \end{cases} 
\eeq
and it is then easy to show that both $g_1$ and $g_2$ satisfy the following RH problem:
\subsubsection*{RH problem for $g$}
\begin{itemize}
\item[(a)] $g:\mathbb{C} \setminus \mathbb{R} \cup_{b \in B} \Gamma_{b} \rightarrow \mathbb{C}$ is analytic. 
\item[(b)] $g$ has jumps,
\begin{align}
g_+(z) &= g_-(z) e^{-2 i \hat{\Lambda}(z) \C{V}_1(z) -2 \pi i \hat{\alpha}_\Gamma(z)}, & z \in & \mathbb{R} \\
g_+(z) &= g_-(z) e^{-2 \pi i \hat{\alpha}_{\Gamma}(z)}, & z \in & \bigcup_{b \in B} \Gamma_b \setminus \mathbb{R}
\end{align}
where $\C{V}_1(z) = \chi_{\mathbb{R}^-}(z)$ in the edge case and  $\C{V}_1(z) = 1$ in the interior case.
\item[(c)] As $z \rightarrow \infty$, $g(z) = \bigO(z^{-\eta})$.
\item[(d)] As $z \rightarrow b \in B$ we have,
\beq
g(z) = \bigO(1) \exp\left[\sum^{\hat{d}_{b}-1}_{j=1}\frac{1}{2j} \tau_{b,j}^j (z-b)^{-j}\right] [z-b]^{-\hat{\alpha}_b}.
\eeq
As $z \rightarrow a \in \partial I$ we have $g(z) = \bigO(1)$.
\end{itemize}
\begin{remark}
In condition (d) we have $g(z) = \bigO(1)$ as $z \rightarrow a \in \partial I$. In principle we could have had $g(z) = \bigO(\log(z -a))$, however it is easy to see by deleting the jumps near $a$ that we must have $\bigO(1)$ behaviour at $a$.
\end{remark}
We now transform $g$ in order to put it into a known form. Define,
\beq
\hat{g}(z) := g(z) z^{-\hat{\alpha}_\mathrm{tot}}\prod_{b\in B} \left([z-b]^{\hat{\alpha}_b} \exp\left[-\sum^{\hat{d}_{b}-1}_{j=1}\frac{1}{2j} \tau_{b,j}^j (z-b)^{-j}\right]\right).
\eeq
We then have,
\subsubsection*{RH problem for $\hat{g}$}
\begin{itemize}
\item[(a)] $\hat{g}:\mathbb{C} \setminus \mathbb{R}$ is analytic. 
\item[(b)] $\hat{g}$ has jumps,
\beq
\hat{g}_+(z) = \hat{g}_-(z) e^{-2 i \hat{\Lambda}(z) \C{V}_1(z) - 2 \pi i \hat{\alpha}_\mathrm{tot} \chi_{\mathbb{R}^-}(z) },\quad z \in \mathbb{R} 
\eeq
\item[(c)] As $z \rightarrow \infty$, $\hat{g}(z) = \bigO(z^{-\eta})$.
\item[(d)] As $z \rightarrow 0$, $\hat{g}(z) = \bigO(z^{-\hat{\alpha}_\mathrm{tot}})$.
\end{itemize}
At this point we note that for the edge case a problem of this form has been solved in \cite{IKO} and for the interior case this form of problem has been solved in \cite{claeys2008multi}. In both cases one finds, via Carlson's theorem, that $g \equiv 0$.  The only difference in the current case compared to \cite{IKO, claeys2008multi} is that here $\hat{\Lambda}$ is not explicit. Nevertheless, the form of proof used in \cite{claeys2008multi} does not rely on the precise form of $\hat{\Lambda}$ and therefore the result follows immediately in the interior case. For the edge case we have essentially the same situation, however we feel it is useful to highlight some of the changes necessary in the proof from \cite{IKO}. We first follow \cite{IKO} by making a change of variable,
\beq
h(z) := \begin{cases}\hat{g}(z^2), & \Re z > 0 \\ \hat{g}(z^2) e^{-2\pi i \hat{\alpha}_\mathrm{tot} \theta(z)} e^{-2\Lambda(z^2)}, & \Re z < 0. \end{cases}
\eeq
It is straight-forward to verify that $h$ is analytic in the region $\Re z > 0$. The difference compared to \cite{IKO} arises in the next transformation,
\beq
\hat{h}(z) := \left(\frac{z}{1+z} \right)^{\frac{4k+8}{2k+3} \alpha} h\left(z^\frac{2k+4}{2k+3} \right).
\eeq
Note that the above equation reduces to equation (2.29) in \cite{IKO} in the case $k=0$. One can then verify that $\hat{h}$ is analytic for $\Re z>0$, is bounded for $\Re z \geq 0$ and for $z \rightarrow \infty$ on the line $\Re z = 0$ we have,
\beq
|\hat{h}(ix)| \leq C e^{-c |x|^{2k+4}}
\eeq
where $C, c > 0$ are constants. By Carlson's theorem we are able to conclude that $\hat{h} \equiv 0$. This completes the proof of the vanishing lemma.
\end{proof}

The proof of Theorem \ref{existence thm} follows from the description of the RH problem in terms of singular integral
equations of Cauchy-type whose corresponding operator is a Fredholm operator of index zero. It can be shown that the kernel of this operator is trivial if and only if the vanishing lemma holds and therefore the vanishing lemma implies the integral equation is solvable. A detailed description of these points can be found in \cite{IKO}. One subtlety of this argument is that the RH problem under consideration must be equivalent to one with no singular points which is not the case here. On this point we follow the same argument given in the existence proof found in \cite{ACM}.

Noting the asymptotic behaviour of $\Phi$ as $z \rightarrow b \in B$, we make the following transformation $\Phi \mapsto \Phi_{b,0}(z)$,
\beq
\Phi(z) = \Phi_{b,0}(z)\exp\left[-\sum^{\hat{d}_{b}-1}_{j=1}\frac{1}{2j} \tau^j_{b,j} (z-b)^{-j}\sigma_3\right][z-b]^{\hat{\alpha}_b \sigma_3},
\eeq
where $\Phi_{b,0}(z)$ is bounded as $z \rightarrow b \in B$. Substituting the above relation into the jump conditions yield that $\Phi_{b,0}$ has no jumps in a neighbourhood of $b$ if $b \notin \mathbb{R}$. This implies that $\Phi_{b,0}$ is analytic near such a point. For $b \in \mathbb{R}$ we instead find that $\Phi_{b,0}$ must satisfy the jump $j_0:=\Phi_{{b,0},-}(z)^{-1}\Phi_{{b,0},+}(z)$,
\begin{align}
j_0 &= \left(
\begin{array}{cc}
 e^{2 \pi i  (\hat{\alpha}_\Gamma(z) - \chi_{\Gamma_b}(z)\hat{\alpha}_b)} & \chi_I(z) e^{-\sum^{\hat{d}_{b}-1}_{j=1}\frac{1}{j} \tau^j_{b,j} (z-b)^{-j}}|z-b|^{2\hat{\alpha}_b} \\
 0 & e^{-2 \pi i (\hat{\alpha}_\Gamma(z) - \chi_{\Gamma_b}(z)\hat{\alpha}_b)} \\
\end{array} \right), &z\in I \cup \bigcup_{b' \in B} \Gamma_{b'}.
\end{align}
We now let,
\beq 
\Phi_{b,0}(z) = \widetilde{\Phi}_{b,0}(z) \begin{pmatrix}1& \tilde{f}(z)\\ 0 & 1\end{pmatrix} e^{\pi i (\hat{\alpha}_\Gamma(z) - \chi_{\Gamma_b}(z)\hat{\alpha}_b) \theta(z) \sigma_3}
\eeq
and note that if we choose $\tilde{f}$ such that,
\beq
\label{tf jump}
\tilde{f}_+(z) = \tilde{f}_-(z) + \chi_I(z) e^{-\sum^{\hat{d}_{b}-1}_{j=1}\frac{1}{j} \tau^j_{b,j} (z-b)^{-j}}|z-b|^{2\hat{\alpha}_b}
\eeq
for $z$ in any fixed interval around $b$ we have that $\widetilde{\Phi}_{b,0}$ has no jumps in a fixed neighbourhood of $b$ and is hence analytic there. It is important to note at this point that a solution of \eqref{tf jump} only exists if $\tau_{b,j}$ is such that $e^{-\sum^{\hat{d}_{b}-1}_{j=1}\frac{1}{j} \tau^j_{b,j} (z-b)^{-j}}|z-b|^{2\hat{\alpha}_b}$ is bounded on $I$. This leads directly to the constraints \eqref{singbound} stated in the theorem.

Let $D_b$ be a fixed disc centered on $b \in B$. We now define,
\beq
\widetilde{\Phi}(z) = \begin{cases} \Phi(z), & \mbox{$z \in \mathbb{C} \setminus \bigcup_{b \in B} D_b $}, \\ \widetilde{\Phi}_{b,0}(z)& \mbox{$\bigcup_{b \in B} D_b$}  \end{cases}
\eeq
and note that $\widetilde{\Phi}$ has no singular points. This completes the proof for the edge and interior case. In the exterior case we postpone a proof until Remark \ref{ext proof}.

\begin{remark}\label{admissible vec}
We have shown that the canonical RH problem is solvable for a certain set of admissible real vectors $\vec{\tau}_\infty$. The transformation to a RH problem with no singular points also lets us use the same argument as used in \cite{claeysPIsol} to show that the canonical RH problem is solvable for complex parameters in a neighbourhood of such an admissible vector. 
\end{remark}

\section{Asymptotic analysis of the RH problem for orthogonal polynomials}
\label{DZsteep}

\subsection{The RH problem for orthogonal polynomials}
An effective way to characterise orthogonal polynomials appearing \eqref{Keqn} is via a well known RH problem due to
Fokas-Its-Kitaev \cite{FokasItsKitaev}.

\subsubsection*{RH problem for $Y$}
\begin{itemize}
\item[(a)] $Y: \mathbb{C}\setminus \C{I} \rightarrow \mathbb{C}^{2 \times 2} $ is analytic.
\item[(b)] The limits of $Y$ as $z$ approaches $\mathbb{R}$ from above and below exist, are continuous on $\mathbb{R}$ and are denoted by $Y_+$ and $Y_-$ respectively. Furthermore they are related by
\beq
Y_+(x) = Y_-(x) \begin{pmatrix}1&w(x)\\0&1\end{pmatrix},\qquad x\in \C{I}.
\eeq
\item[(c)] $Y(z) = (I+\bigO(z^{-1}))z^{n\sigma_3}$ as $z \rightarrow \infty$.
\item[(d)] As $z \rightarrow a$ for $a \in \C{A} \setminus \C{B}$,
\beq
\label{hard edge asymp}
Y(z) =  Y_a(z) \times \begin{cases} e^{\frac{1}{2\pi i } w(z) \log (z-a) \sigma_+}, & \mbox{if $a$ is a right edge of an interval,}\\e^{-\frac{1}{2\pi i } w(z) \log (a-z) \sigma_+}, & \mbox{if $a$ is a left edge of an interval.} \end{cases}
\eeq
where $Y_a$ is analytic in a neighbourhood of $a$.
As $z \rightarrow b$ for $b \in \C{B}$ we have,
\beq
Y(z)= \bigO\begin{pmatrix}1 & 1 \\ 1 & 1 \end{pmatrix}.
\eeq
\end{itemize}
\begin{remark}
The behaviour \eqref{hard edge asymp} follows from the fact that $Y$ has at most a $\log$ divergence at $a$ and $Y$ has the same jumps as $\exp(\frac{1}{2\pi i }w(z) \log (z-a) \sigma_+)$ at a right edge. A similar statement holds at a left edge.
\end{remark}
This RH problem has a unique solution,
\beq
\label{Ysol}
Y(z) = \begin{pmatrix}p_n(z)&q_n(z)\\ -\frac{2\pi i}{h_{n-1}} p_{n-1}(z)& -\frac{2\pi i}{h_{n-1}} q_{n-1}(z)\end{pmatrix},
\eeq
where $p_j$ is the degree $j$ monic orthogonal polynomial defined in \eqref{ortho p} and
\beq
q_j(z) := \frac{1}{2\pi i} \int_{\C{I}} \frac{p_j(x) w(x)}{x-z} dx.
\eeq

\subsection{Transformation to constant jumps}

To construct the first transformation we introduce an analytic continuation of $w_\mathrm{br}$.
\bdf
We define,
\beq
\bar{w}_\mathrm{br}(z):= \prod_{b\in \C{B}} [z-b]^{2\alpha_b}. 
\eeq
Let $\bar{w}(z) := w_\mathrm{reg}(z)w_\mathrm{sing}(z)\bar{w}_\mathrm{br}(z)$.
\edf
\begin{remark}
We have that $\bar{w}_\mathrm{br}$ has the following jump properties,
\begin{align}
\bar{w}_\mathrm{br}(z)_+ &= \bar{w}_\mathrm{br}(z)_- e^{4\pi i \alpha_\Gamma}, & z \in  \bigcup_{b \in {\C{B}}} \Gamma_{b}  \\
\bar{w}_\mathrm{br}(z)_+ \bar{w}_\mathrm{br}(z)_- &= \prod_{b\in \C{B}} |z-b|^{4\alpha_b} & z \in \mathbb{R}.
\end{align}
Note that the second jump property relies on the fact that $\C{B}$ contains conjugate pairs of points.
\end{remark}

\bdf
\label{alpha def}
Define $\alpha_{\Gamma}(z)$ in an identical way to Definition \ref{hat alpha def} with $\C{B}$ replacing $B$ and $\alpha_b$ replacing $\hat{\alpha}_b$. \edf

Defining $\Psi(z) := Y(z) \bar{w}(z)^\frac{\sigma_3}{2}$, we have,
\subsubsection*{RH Problem for $\Psi$}
\begin{itemize}
\item[(a)] $\Psi: \mathbb{C} \setminus \C{I} \setminus \bigcup_{b \in {\C{B}}} \Gamma_{{b}} \rightarrow \mathbb{C}^{2 \times 2}$ is analytic.
\item[(b)] Let $j_\Psi(z) := \Psi_-(z)^{-1}\Psi_+(z)$. Then,
\begin{align}
j_\Psi(z) &=\left(
\begin{array}{cc}
 e^{2 i \pi  \alpha_\Gamma(z)} & \chi_{\C{I}}(z) \\
 0 & e^{-2 i \pi  \alpha_\Gamma(z)} \\
\end{array}
\right), &z\in \C{I} \cup \bigcup_{b \in {\C{B}}} \Gamma_{{b}}.
\end{align}
\item[(c)] $\Psi(z) = (I+\bigO(z^{-1}))z^{(n+\alpha_\mathrm{tot})\sigma_3}w_{\mathrm{reg}}(z)^\frac{\sigma_3}{2}$ as $z \rightarrow \infty$, where $\alpha_\mathrm{tot} := \sum_{b \in \C{B}} \alpha_b$.
\item[(d)] As $z \to b \in \C{B}$,
\beq
\Psi(z) = \bigO(1)\exp\left[-\sum^{d_{b}-1}_{j=1}\frac{1}{2j} t_{b,j} (z-b)^{-j}\sigma_3\right] [z-b]^{\alpha_b  \sigma_3}
\eeq
As $z \to a \in \C{A} \setminus \C B $
\beq
\Psi(z) = \Psi_{a,0}(z)\times \begin{cases} e^{\frac{1}{2\pi i } \log (z-a) \sigma_+}, & \mbox{if $a$ is a right edge of an interval,}\\e^{-\frac{1}{2\pi i }\log (a-z) \sigma_+}, & \mbox{if $a$ is a left edge of an interval.} \end{cases}
\eeq
Here $\Psi_{a,0}$ is analytic with respect to $z$ at $a$.
\end{itemize}

\begin{remark}\label{ext proof}
Observe that the RH problem for $\Psi$ exactly matches the definition of the canonical model problem of exterior type and we therefore are able to construct a solution to a model problem in terms of orthogonal polynomials.
\end{remark}
\subsection{Opening the lens}
We now perform the standard step of opening the lens. The difference here is that,
since the jumps on the lens contours are constant, the lens contours are unconstrained.
In a subsequent transformation we will use the $g$-function to normalise at infinity, at which
point the lens contours will be required to stay within a region in which they converge to the
identity as $n \rightarrow \infty$.

Due to the rather general nature of the problem we need some additional definitions in order
to define the lens contours and associated regions.

\subsubsection{Definitions of contours}

\begin{definition}
\label{lens contours}
For an interval $\sigma$ of $\mathbb{R}$ define the contours $\Sigma^{\sigma}_\pm$ to be smooth contours from $\inf \sigma$ to $\sup \sigma$ in the regions $\sigma \pm i \mathbb{R}^+$ respectively. A {\em{lens contour for}} $\sigma$ is the contour $\Sigma^\sigma_+\cup \Sigma^\sigma_-$ with all contours orientated from $\inf \sigma$ to $\sup \sigma$ .
\end{definition}

\bdf
\label{lens regions}
Given a lens contour for an interval $\sigma$ define $\Omega^{\sigma}_\pm$ to be the region of $\mathbb{C}$ bounded by $\sigma\cup \Sigma^\sigma_\pm$ such that the $+$ side and $-$ side of $\sigma$ is contained in $\Omega^{\sigma}_+$ and $\Omega^{\sigma}_-$ respectively.
\edf

\bdf
\label{collection of lens}
For a collection of pairwise disjoint intervals $\C{I}$ define a {\em{full lens contour}} $\Sigma^{\C{I}}$ as all intervals in $\C{I}$ treated as contours together with their lens contours. Denote the lens contours of $\Sigma^{\C{I}}$ by $\Sigma^{\C{I}}_\pm$ and define $\Omega^\C{I}_\pm := \bigcup_{\sigma \in \C{I}} \Omega^{\sigma}_\pm$.
\edf

\begin{figure}[t]
\centering 
\includegraphics[scale=0.4]{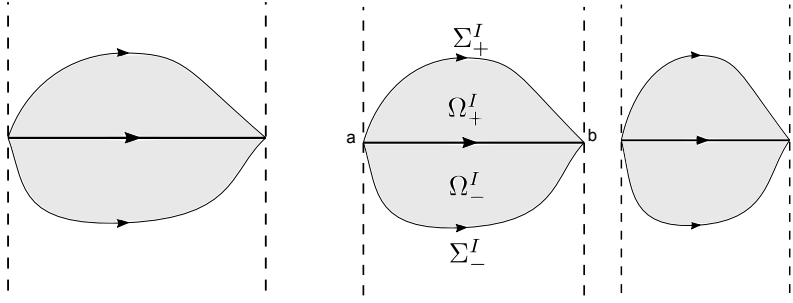}
\caption{An example of contours constructed using Definitions \ref{lens contours}, \ref{lens regions} and \ref{collection of lens}.}
\label{contourdefs}
\end{figure}

Finally we surround all singular points with discs in which we will later construct the local parametricies.

\begin{definition}
For all $x_* \in \C{P}_*\cup\C{R}_*$ define, $D_{x_*}(\delta)$ to be a disc of diameter delta centered at $x_*$ whose
boundary is orientated clockwise. Let $U = \bigcup_{x_*\in \C{P}_*\cup\C{R}_*} D_{x_*}(\delta_{x_*})$. We will always work
with $n$ large enough that the discs contain all points in $\C{P} \cup \C{R}$.
\end{definition}

Let $\bar{\C{S}} \subset \C{S} \cap U$ be the intervals of $\C{S}$ which are fully contained in $U$, i.e. they don't intersect $\partial U$. Let $\gamma = \C{S} \setminus \bar{\C{S}} \setminus \C{P}_*$. We now open the lens of $\gamma$ by defining,
\beq
S(z) := \Psi(z) K(z),
\eeq
where $K$ is a piecewise function designed to open the lens,
\beq
K(z):= \left\{
\begin{array}{lr}
I,  & \mbox{for } z \in \mathbb{C} \setminus \Omega^{\gamma}_+ \setminus \Omega^{\gamma}_-, \\
\begin{pmatrix}1&0\\-e^{2\pi i \alpha_\Gamma(z)}&1\end{pmatrix}, & \mbox{for } z \in \Omega^{\gamma}_+,\\
\begin{pmatrix}1&0\\e^{-2\pi i \alpha_\Gamma(z)}&1\end{pmatrix}, & \mbox{for } z \in \Omega^{\gamma}_-.\\
\end{array}
\right.
\eeq
The function $S$ satisfies the following RH problem.
\subsubsection*{RH Problem for $S$}
\begin{itemize}
\item[(a)] $S: \mathbb{C} \setminus \C{I}  \setminus \Sigma^\gamma \setminus \bigcup_{b \in {\C{B}}} \Gamma_{{b}} \rightarrow \mathbb{C}^{2 \times 2}$ is analytic.
\item[(b)] Let $j_S(z) := S_-(z)^{-1}S_+(z)$ then,
\begin{align}
j_S &= \left(
\begin{array}{cc}
 0 & 1 \\
 -1 & 0 \\
\end{array}
\right), &z\in \gamma,\\
j_S &= \left(
\begin{array}{cc}
 e^{2 i \pi  \alpha_\Gamma(z)} & \chi_{\C{I}}(z) \\
 0 & e^{-2 i \pi  \alpha_\Gamma(z)} \\
\end{array}
\right), &z\in \C{I} \cup \bigcup_{b \in {\C{B}}} \Gamma_{{b}} \setminus \gamma,\\
j_S &= \left(
\begin{array}{cc}
 1 & 0 \\
 e^{\pm 2 i \pi  \alpha_\Gamma(z)} & 1 \\
\end{array}
\right), &z \in \Sigma^\gamma_\pm.
\end{align}
\item[(c)] $S(z)$ has the same asymptotic behaviour as $\Psi(z)$ as $z \rightarrow \infty$.
\item[(d)] $S(z)$ has the same asymptotic behaviour as $\Psi(z)$ as $z \rightarrow b$ for $b \in \C{A} \cup \C{B}$ unless
we also have $b \in \C{S}$, in which case,
\beq
S(z) = \Psi(z) \times \begin{cases}
I  & \mbox{for } z \in \mathbb{C} \setminus \Omega^{\gamma}_+ \setminus \Omega^{\gamma}_- \\[1em]
\begin{pmatrix}1&0\\-e^{2\pi i \alpha_\Gamma(z)}&1\end{pmatrix}, & \mbox{for } z \in \Omega^{\gamma}_+\\[1em]
\begin{pmatrix}1&0\\e^{-2\pi i \alpha_\Gamma(z)}&1\end{pmatrix}, & \mbox{for } z \in \Omega^{\gamma}_-\\
\end{cases}
\eeq

\end{itemize}

\subsection{Normalisation at infinity}

The next transformation takes the form,
\beq
T(z) := e^{-\frac{n l \sigma_3}{2}} S(z) \times \begin{cases} e^{-n\xi(z) \sigma_3},&\mbox{if $z \in \mathbb{C} \setminus U,$}\\
I  ,&\mbox{if $z \in U.$ }\end{cases}
\eeq
The above transformation has the effect of normalising the problem at infinity.

\subsubsection*{RH Problem for $T$}
\begin{itemize}
\item[(a)] $T: \mathbb{C} \setminus \C{I} \setminus \Sigma^\gamma \setminus \partial U \setminus \bigcup_{b \in {\C{B}}} \Gamma_{{b}} \rightarrow \mathbb{C}^{2 \times 2}$ is analytic.
\item[(b)] Let $j_T(z) := T_-(z)^{-1}T_+(z)$. Then,
\begin{align}
j_T &= j_S, & z &\in (\C{I}  \setminus \Sigma^\gamma \setminus \bigcup_{b \in {\C{B}}} \Gamma_{{b}}) \cap U,\\ 
j_T &= \left(
\begin{array}{cc}
 0 & 1 \\
 -1 & 0 \\
\end{array}
\right), &z &\in \gamma \setminus U,\\
j_T &= \left(
\begin{array}{cc}
e^{2 \pi i(\alpha_\Gamma(z)- n \epsilon(z))}& \chi_{\C{I}}(z)e^{n \left(\xi_+(z)+\xi_-(z)\right)} \\
 0 & e^{-2 \pi i (\alpha_\Gamma(z)-n \epsilon(z))} \\
\end{array}
\right), &z&\in \C{I} \cup \bigcup_{b \in {\C{B}}} \Gamma_{{b}} \setminus \gamma\setminus U,\\
j_T &= \left(
\begin{array}{cc}
 1 & 0 \\
 e^{\pm 2 \pi i  \alpha_\Gamma z-2 n \xi(z)} & 1 \\
\end{array}
\right), &z &\in \Sigma^\gamma_\pm \setminus U\\.
j_T &= e^{-n  \xi(z) \sigma_3}, & z &\in \partial U.
\end{align}

\item[(c)] As $z \rightarrow \infty$,
\beq
T(z) = (1+ \bigO(z^{-1})) z^{\alpha_\mathrm{tot}}.
\eeq
\item[(d)] As $z \rightarrow b \in \C{A} \cup \C{B}$, $T(z)$ has the same asymptotics as $ e^{-\frac{n l \sigma_3}{2}}S(z)$. This is because for sufficiently large $n$ all points in $\C{A} \cup \C{B}$ are in $U$.
\end{itemize}

\subsection{Final transformation}
For the final transformation we need a global approximation to the RH problem for $T$. To define this we first need to introduce some new objects.
\bdf
Define $\epsilon^\pm_{x_*} := \epsilon(x_*\pm \frac{1}{2}\delta_{x_*})$ and $\alpha^\pm_{x_*} := \alpha_\Gamma(x_*\pm \frac{1}{2}\delta_{x_*})$; these are the values of $\epsilon$ and $\alpha_\Gamma$ at the points at which $\partial D_{x_*}$, intersects $\mathbb{R}$. We also define $\epsilon^\delta_{x_*} := \epsilon^-_{x_*} -  \epsilon^+_{x_*}$ and $\alpha^\delta_{x_*} := \alpha^-_{x_*} -  \alpha^+_{x_*}$. Note that $\epsilon^\pm_{x_*}$ can be undefined when $x_*$ is of edge or interior type.
\edf

\bdf \label{tilde alpha def}
Let,
\beq
\tilde{\alpha}_\Gamma(z) = \begin{cases}\alpha_\Gamma(z), & \mbox{for $z \in \mathbb{R} \setminus U$,}\\
\alpha^+_{x_*}, & \mbox{for $z \in D_{x_*} \cap (x_*,\infty)$,}\\
\alpha^-_{x_*}, & \mbox{for $z \in D_{x_*} \cap (-\infty,x_*)$.} \end{cases}
\eeq
\edf

\bdf \label{tilde epsilon def}
Let,
\beq
\tilde{\epsilon}_1(z) = \begin{cases}\epsilon(z), & \mbox{for $z \in \mathbb{R} \setminus \C{S}_* \setminus U$,}\\
\epsilon^+_{x_*}, & \mbox{for $z \in D_{x_*} \cap (x_*,\infty)$ with $x_*$ an edge or exterior point,}\\
\epsilon^-_{x_*}, & \mbox{for $z \in D_{x_*} \cap (-\infty,x_*)$ with $x_*$ an exterior point.} \end{cases}
\eeq

Note that $\tilde{\epsilon}_1$ is constant on connected components of $\mathbb{R} \setminus \C{J}_*$. Now we define $\tilde{\epsilon}(z) := \tilde{\epsilon}_1(z) - n^{-1}\C{N}(z)$ where $\C{N} : \mathbb{R}\setminus \C{J}_* \rightarrow \mathbb{N}^0$ is a function which is constant on each connected component of $\mathbb{R}\setminus \C{J}_*$. We chose $\C{N}(z)$ such that when $x_*$ is an exterior point we have $ n|\tilde{\epsilon}^-_{x_*} -  \tilde{\epsilon}^+_{x_*}|< \frac{1}{2}$ and $\tilde{\epsilon}(z) = 0$ if $n \tilde{\epsilon}_1(z) \in \mathbb{Z}$.
\edf

\bdf
Define $\tilde{\epsilon}^\pm_{x_*} := \tilde{\epsilon}(x_*\pm \frac{1}{2}\delta_{x_*})$ and $\tilde{\epsilon}^\delta_{x_*} := \tilde{\epsilon}^-_{x_*} -  \tilde{\epsilon}^+_{x_*}$.
\edf

\begin{figure}[t]
\centering 
\includegraphics[scale=0.6]{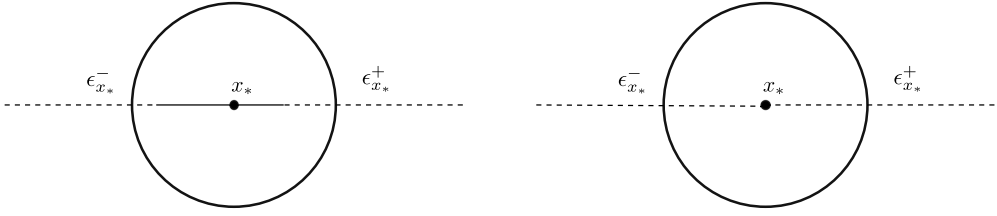}
\caption{On the left the figure shows the behaviour of $\epsilon$ on the connected components of $\mathbb{R} \setminus \C{J}$ (shown as dashed lines) in the case that $x_*$ is an exterior point. The value of $\epsilon$ at the boundary of the disc are labelled $\epsilon^\pm_{x_*}$. In the figure on the right we show the behaviour of $\tilde{\epsilon}_1$ introduced in Definition \ref{tilde epsilon def}, in which the values of $\epsilon$ on the boundary of the disc are extended all the way to $x_*$.}
\label{epsilontildedef}
\end{figure}

\begin{remark}
The idea in the definition of $\tilde{\alpha}_\Gamma$ and $\tilde{\epsilon}$ is that we want them to match $\alpha_\Gamma$ and $\epsilon$ where they are defined but extend them to the domain obtained by shrinking the discs in $U$ to points (see Figure \ref{epsilontildedef}). The reason for this is that we can then use a simpler global parametrix built on a Riemann surface which is not degenerate as $n \rightarrow \infty$. Furthermore, in the definition of $\tilde{\epsilon}$ we have used the freedom that the jumps of $T$ are invariant under $n\epsilon \mapsto n\epsilon + k$ where $k$ is integer, to minimise the difference in $\tilde{\epsilon}$ across singular points $x_*$. 
\end{remark}

Finally we define the {\em{global parametrix}} as a function $P^{(\infty)}$ satisfying the following RH problem,
\subsubsection*{RH problem for $P^{(\infty)}$}
\begin{itemize}
\item[(a)] $P^{(\infty)}: \mathbb{C} \setminus (-\infty,\sup (\C{B}_*\cup \C{J}_*)] \rightarrow \mathbb{C}^{2 \times 2}$ is analytic.
\item[(b)] $P^{(\infty)}$ has the jump relations
\begin{align}
&P^{(\infty)}_+(z) = P^{(\infty)}_-(z) \begin{pmatrix}0 &1 \\ -1 &0\end{pmatrix},&  z \in \C{S}_*\\
&P^{(\infty)}_+(z) = P^{(\infty)}_-(z) e^{2\pi i(\tilde{\alpha}_\Gamma(z) - n \tilde{\epsilon}(z)) \sigma_3} ,&  z \in \mathbb{R} \setminus \C{J}_*,
\end{align}
Recall that $\tilde{\alpha}_\Gamma : \mathbb{R}\setminus \C{J}_* \rightarrow \mathbb{R}$ and $\tilde{\epsilon} : \mathbb{R}\setminus \C{J}_* \rightarrow \mathbb{R}$ are functions which are constant on each connected component of $\mathbb{R}\setminus \C{J}_*$. Furthermore $\tilde{\epsilon}$ is zero on infinite intervals and $\tilde{\alpha}_\Gamma$ is zero on $\mathbb{R} \setminus (-\infty,\sup \C{B}_*)$.
\item[(c)] As $z \rightarrow \infty$,
\beq
P^{(\infty)}(z) = (I + \bigO(z^{-1}))z^{\alpha_\mathrm{tot} \sigma_3}.
\eeq
\item[(d)] As $z \rightarrow \hat{z} \in \partial \C{S}_*$
\beq
P^{(\infty)}(z) = \bigO((z-\hat{z})^{-\frac{1}{4}}).
\eeq
As $z \rightarrow x_* \in \C{E}_*$
\beq
\label{z to xstar}
P^{(\infty)}(z) = \bigO(1) (z-\hat{z})^{(\alpha^\delta_{x_*} - n \tilde{\epsilon}^\delta_{x_*})\sigma_3}.
\eeq
\end{itemize}

We now need to define a {\em{local parametrix}} $P^{(x_*)}$ for $x_* \in \C{P}_*\cup \C{R}_*$. We define $P^{(x_*)}$ as the solution to the following RH problem,
\subsubsection*{RH problem for $\C{P}^{(x_*)}$}
\label{local parametrix}
\begin{itemize}
\item[(a)] $P^{(x_*)} : D_{x_*}(\delta_{x_*}) \setminus \C{I} \setminus \Sigma^\gamma \setminus \bigcup_{b \in {\C{B}}} \Gamma_{{b}}  \rightarrow \mathbb{C}^{2 \times 2}$ is analytic.
\item[(b)] Let $j_{P^{(x_*)}} := P^{(x_*)}_-(z)^{-1}P^{(x_*)}_+(z)$ then $j_{P^{(x_*)}} = j_T$.
\item[(c)] As $z \rightarrow b \in \C{B}|_{x_*}$, $P^{(x_*)}(z) = T(z)$ where the equality is understood to mean equality between asymptotics series and the notation $\C{B}|_{x_*}$ was introduced in Definition \ref{bar x_star def}.
\item[(d)] As $n\rightarrow \infty$, $P^{(x_*)}(z) = (I + o(1)) P^{(\infty)}(z) e^{n \xi(z)\sigma_3}$ uniformly for $z \in \partial D_{x_*}(\delta_{x_*})$.
\end{itemize}

We now define $R(z)$ as,
\begin{equation}
R(z)=\begin{cases}
T(z)P^{(\infty)}(z)^{-1},&\mbox{ for $z \in \mathbb{C} \setminus U$, }\\
T(z)P^{(x_*)}(z)^{-1},&\mbox{ for $z\in D_{x_*}(\delta_{x_*})$, where $x_* \in \C{P}_*\cup \C{R}_*$.}
\end{cases}
\end{equation}
Using the above definition together with Definition \ref{tilde epsilon def} we can derive the following RH problem for $R$.
\subsubsection*{RH problem for $R$}

\begin{itemize}
\item[(a)] $R: \mathbb{C} \setminus \C{I} \setminus (\Sigma^\gamma \setminus U) \setminus  \partial U \rightarrow \mathbb{C}^{2 \times 2}$ is analytic.

\item[(b)] Let $j_R(z) := R_-(z)^{-1}R_+(z)$ then,
\begin{align}
\label{jumps of R: lens}
j_R(z) &= P^{(\infty)}(z) \begin{pmatrix} 1 & 0 \\
 e^{-2 n \xi(z)\pm 2\pi i \alpha_\Gamma(z)} & 1 \end{pmatrix} P^{(\infty)}(z)^{-1}, &z \in \Sigma^\gamma_\pm\setminus U,\\ 
\label{jumps of R: gamma}
j_R(z) &=  P^{(\infty)}_+(z) \begin{pmatrix}  1 & e^{n \left(\xi_-(z)+\xi_+(z)\right) +2 \pi i n \epsilon (z)-2\pi i \alpha_\Gamma(z) }  \\ 0 & 1\end{pmatrix} P^{(\infty)}_+(z)^{-1}, &z\in \C{I} \setminus \gamma \setminus U ,\\ 
\label{jumps of R: disc}
j_R(z) &= P^{(x_*)}_-(z) \begin{pmatrix} e^{-n \xi(z)} & 0 \\
 0 & e^{n \xi(z)} \end{pmatrix} P^{(\infty)}_+(z)^{-1}, &z\in\partial D_{x_*},
\end{align}
\item[(c)] As $z \rightarrow \infty$, $R(z) = I + \bigO(z^{-1})$.
\end{itemize}

\subsection{Proof that $R$ is a small norm RH problem}
To prove that $R$ is a small norm RH problem we will show that each jump matrix of $R$ tends towards
the identity as $n\rightarrow \infty$. The first observation is that for sufficiently large $n$ and any fixed point $x \in \C{S} \setminus U$ we have that $\rho(x) > 0$ \cite{DKMVZ2}. By standard arguments this implies that there exists a choice of lens contours $\Sigma^\gamma_\pm$ such that $\Re \xi(z) > 0$ for all $z \in \Sigma^\gamma_\pm \setminus U$ which together with the form of the jump matrices implies the jump matrices converge uniformly to the identity on $\Sigma^\gamma_\pm \setminus U$. 

 We now turn our attention to \eqref{jumps of R: gamma}. In this case we make use of the variational condition that $\xi_+(x) + \xi_-(x) < 0$ for $z \in \mathbb{R} \setminus \C{S}\setminus U$ to again conclude that the jump matrices converge uniformly to the identity on $\mathbb{R} \setminus \C{S}\setminus U$. 

Finally we turn our attention to the jump on the disc boundary \eqref{jumps of R: disc}. That this jump tends to the identity as $n \rightarrow \infty$ follows trivially from property (d) of the local parametrix.

We therefore conclude that all the jumps for $R$ decay to the identity as $n \rightarrow \infty$.

\section{Construction of global and local parametrices}\label{section: 3}

\subsection{Construction of $P^{(\infty)}$}

Given a collection of closed pairwise disjoint intervals $I$ we first define a scalar Szeg\"{o} function $D(z;\alpha,I)$ that solves the following RH problem,
\subsubsection*{RH problem for $D$}
\begin{itemize}
\item[(a)] $D(z;\alpha,I)$ is non-zero and analytic on $\mathbb{C}\setminus \mathbb{R}$ with respect to $z$.
\item[(b)] $D(z;\alpha,I)$ satisfies the following jump relations:
\begin{align}
& D_+(x;\alpha,I)D_-(x;\alpha,I)=|x|^{2\alpha},& \qquad\mbox{for $x\in I$,} \\
& D_+(x;\alpha,I)=e^{2\pi i \eta_\alpha(x)} D_-(x;\alpha,I),& \qquad\mbox{for $x\in \mathbb{R} \setminus I$,}
\end{align}
where the function $\eta_\alpha : \mathbb{R} \setminus I  \rightarrow \mathbb{R}$ is constant on each connected component of $\mathbb{R} \setminus I$.
\item[(c)] $D$ and $D^{-1}$ remain bounded as $z \rightarrow \hat{z} \in \partial I$ and
\beq
D_\infty := \lim_{z\to\infty}D(z;\alpha,I)
\eeq
exists and is non-zero.
\item[(d)] As $z \rightarrow 0$ we have,
\beq
D(z) = \begin{cases} z^\alpha(1+\bigO(z)), &\mbox{if $0\in I$} \\ \bigO(1), &\mbox{if $0\notin I$} \end{cases}.
\eeq
\end{itemize}
\begin{remark}
The function $\eta_\alpha$ in the definition of $D$ is not arbitrary and is fully determined by the requirement (c) in the RH problem. It is shown in \cite{Vanlassen0305044} how (c) determines $\eta_\alpha$ explicitly, together with an explicit construction of $D$.
\end{remark}

\bdf
We now define two auxiliary functions built using $D$;
\begin{align}
\hat{D}(z;x_*,\alpha) &:= (z-x_*)^{-\alpha} D(z-x_*;\alpha,\C{S}_*-x_*)\\
\bar{D}(z;x_*,\alpha) &:= \frac{\hat{D}(z;b_0,\alpha)}{\hat{D}(z;x_*,\alpha)}, \qquad x_* \notin \C{S}_*,
\end{align}
where $b_0 := \sup(\C{S}_*\cap (-\infty,x_*))$ and principal branches are taken for multi-valued functions. The notation $\C{S}_*-x_*$ means the intervals of $\C{S}_*$ shifted to the right by $x_*$. 
\edf

It is straightforward to verify that $\hat{D}$ and $\bar{D}$ satisfy the following RH problems,

\subsubsection*{RH problem for $\hat{D}$}
\begin{itemize}
\item[(a)] The function $\hat{D}:\mathbb{C}\setminus \mathbb{R} \rightarrow \mathbb{C}$ is analytic in $z$.
\item[(b)] $\hat{D}$ has jumps,
\begin{align}
&\hat{D}_+(z)\hat{D}_-(z) = 1 &  z \in \C{S}_*\\
&\hat{D}_+(z) = \hat{D}_-(z) e^{2\pi i \eta(z)} ,&  z \in (x_*,\infty) \cap (\mathbb{R} \setminus \C{S}_*),\\
&\hat{D}_+(z) = \hat{D}_-(z) e^{2\pi i (\eta(z)-\alpha)} ,&  z \in (-\infty,x_*) \cap (\mathbb{R} \setminus \C{S}_*),
\end{align}
where $\eta: \mathbb{R} \setminus \C{S}_* \rightarrow \mathbb{R}$ is a function constant on each connected component of $\mathbb{R} \setminus \C{S}_*$ and zero on infinite intervals.
\item[(c)] As $z\rightarrow \infty$,
\beq
\hat{D}(z) = z^{-\alpha} D_\infty(1+o(1)).
\eeq
\item[(d)] As $z \rightarrow x_*$,
\beq
\hat{D}(z) = \begin{cases} (z-x_*)^{-\alpha} \bigO(1), & \mbox{if $x_* \notin \C{S}_*$} \\  \bigO(1) & \mbox{if $x_* \in \C{S}_*$}  \end{cases}
\eeq
\end{itemize}
\subsubsection*{RH problem for $\bar{D}$}
\begin{itemize}
\item[(a)] The function $\bar{D}:\mathbb{C}\setminus \mathbb{R} \rightarrow \mathbb{C}$ is analytic.
\item[(b)] $\bar{D}$ has jumps,
\begin{align}
&\bar{D}_+(z)\bar{D}_-(z) = 1 &  z \in \C{S}_*\\
&\bar{D}_+(z) = \bar{D}_-(z) e^{2\pi i \eta(z)} ,&  z \in  \mathbb{R} \setminus \C{S}_* \setminus [b_0,x_*],\\
&\bar{D}_+(z) = \bar{D}_-(z) e^{2\pi i (\eta(z) + \alpha)} ,&  z \in [b_0,x_*] \cap (\mathbb{R} \setminus \C{S}_*),
\end{align}
where $\eta: \mathbb{R} \setminus \C{S}_* \rightarrow \mathbb{R}$ is a function constant on each connected component of $\mathbb{R} \setminus \C{S}_*$ and zero on infinite intervals.
\item[(c)] As $z\rightarrow \infty$,
\beq
\bar{D}(z) = \bar{D}_\infty(1+\bigO(z^{-1})).
\eeq
\item[(d)] As $z \rightarrow x_*$,
\beq
\bar{D}(z) = (z-x_*)^{\alpha} \bigO(1).
\eeq
As $z \rightarrow b_0$,
\beq
\bar{D}(z) = \bigO(1).
\eeq
Note that the above behaviour is obtained under the constraint appearing in the definition of $\bar{D}$ that $x_* \notin \C{S}_*$.
\end{itemize}
We now seek $P^{(\infty)}$ in the form,
\begin{align}
P^{(\infty)}(z) :=&  \prod_{x_* \in \C{E}_*} \bar{D}_\infty(x_*,n\epsilon^\delta_{x_*})^{-\sigma_3}  \prod_{b \in \C{B}_*} \hat{D}_\infty(b,\alpha_b)^{-\sigma_3} \times \nn \\
&\tilde{P}^{(\infty)}(z)  \prod_{b \in \C{B}_*} \hat{D}(z;b,\alpha_b)^{\sigma_3} \prod_{x_* \in \C{E}_*} \bar{D}(z;x_*, n\tilde{\epsilon}^\delta_{x_*})^{\sigma_3}.
\end{align}
Note that all points in $\C{B}_*$ and $\C{E}_*$ are real. Furthermore, using the definition of $\tilde{\alpha}_\Gamma(z)$ (Definition \ref{tilde alpha def}) together with the fact that $\tilde{\epsilon} - \sum_{x_* \in \C{E}_*} \tilde{\epsilon}^\delta_{x_*} \chi_{[b_0, x_*]}(z)$ is constant on connected components of $\mathbb{R} \setminus \C{J}_*$ we obtain that the function $\tilde{P}^{(\infty)}$ must solve a RH problem of the form,
\subsubsection*{RH problem for $\tilde P^{(\infty)}$}
\begin{itemize}
\item[(a)] $\tilde P^{(\infty)}:\mathbb{C}\setminus \mathbb{R} \to\mathbb{C}^{2\times 2}$ is analytic in $z$.
\item[(b)] $\tilde P^{(\infty)}$ has the following jump relations:
\beq
\tilde P^{(\infty)}_+(x)=\tilde P^{(\infty)}_-(x)\begin{pmatrix}0 & 1 \\ -1 & 0 \end{pmatrix},\qquad\mbox{for $x\in \C{S}_*$,}
\eeq
\beq
\tilde P^{(\infty)}_+(x)=\tilde P^{(\infty)}_-(x)
\begin{pmatrix}
e^{2\pi i \upsilon(z)} & 0 \\
0 & e^{-2\pi i\upsilon(z)}
\end{pmatrix}, \qquad \mbox{for $x\in \mathbb{R} \setminus \C S_*$.}
\eeq
where $\upsilon(z): \mathbb{R} \setminus \C{S}_* \rightarrow \mathbb{R}$ is a function constant on each connected component of $\C{S}_* \setminus \mathbb{R}$ and zero on infinite intervals.
\item[(c)] As $z \rightarrow \infty $, $\tilde P^{(\infty)}(z)=I+O(1/z)$.
\end{itemize}
The solution to the RH problem for $\tilde P^{(\infty)}$ in terms of theta functions is well known and can be found in \cite{Vanlassen0305044}.

\subsection{Asymptotic behaviour of $\xi$ and $P^{(\infty)}$ near singular points}
Before constructing the local parametrix we first prove some useful lemmas regarding the behaviour of $P^{(\infty)}$ and $\xi$ near $x_*$. We will make heavy use of these lemmas in the construction of $P^{(x_*)}$. In the following we consider the case of $x_*$ being an interior, exterior or right edge point. The case of a left edge point we omit as it is similar to the right edge case.

\begin{lemma}
\label{Pinf asymptotics}
Consider $P^{(\infty)}(z)$ inside the disc $D_{x_*}$. It will be convenient to change to the variable $z = x_*+ \zeta \in \partial D_{x_*}$. Let $E(\zeta)$ be an analytic function in a fixed neighbourhood of $\zeta = 0$. 

\begin{itemize}
\item[(i)] If $x_* \in \C{S}_*$, i.e. an interior point, then we have,
\beq
P^{(\infty)}(x_* + \zeta) = E(\zeta) Q(\zeta),
\eeq
where we have used $Q$ as defined in \eqref{Q def}.
\item[(ii)] If $x_* \in \partial\C{S}_*$ with $x_*$ forming the right edge of an interval in $\partial\C{S}_*$, i.e. a right edge point, then we have,
\beq
P^{(\infty)}(x_* + \zeta) =  E(\zeta) \zeta^{-\frac{\sigma_3}{4}}N  e^{\pi i (\alpha^+_{x_*} - n \epsilon^+_{x_*})\theta(\zeta) \sigma_3},
\eeq
where we recall $N := 2^{-\frac{1}{2}}(I + i \sigma_1)$.
\item[(iii)] If $x_* \in \C{E}_*$ then we have,
\beq
P^{(\infty)}(x_* + \zeta) =  E(\zeta)e^{\pi i (\alpha^+_{x_*} - n \epsilon^+_{x_*})\theta(\zeta) \sigma_3} \zeta^{\beta \sigma_3},
\eeq
where we have introduced $\beta := \alpha^\delta_{x_*} - n\tilde{\epsilon}^\delta_{x_*}$.
\end{itemize}
\end{lemma}

\begin{proof}
We consider each case in turn.
\begin{itemize}
\item[(i)] If $x_* \in \C{S}_*$ then $E(z) := P^{(\infty)}(z) Q(z)^{-1}$ has no jump in a neighbourhood of $x_*$. Furthermore $E(z)$ is bounded at $x_*$ and therefore we conclude $E$ is an analytic function in this neighbourhood. Redefining $E$ by shifting the argument we obtain the result (i).
\item[(ii)] If $x_* \in \partial\C{S}_*$ with $x_*$ forming the right edge of an interval in $\partial\C{S}_*$, then,
\beq
E(z) := P^{(\infty)}(z)  e^{-\pi i (\alpha^+_{x_*} - n \epsilon^+_{x_*})\theta(\zeta) \sigma_3} N^{-1}(z-x_*)^\frac{\sigma_3}{4}
\eeq has no jumps in a neighbourhood $M$ of $x_*$ and therefore $E$ is analytic in $M \setminus \{x_*\}$. Furthermore, given that the behaviour of $P^{(\infty)}(z)$ is $\bigO((z-x_*)^{-\frac{1}{4}})$ near $x_*$ we see that $E$ can be analytically continued to all of $M$. Redefining $E$ by shifting the argument we obtain the result (ii).
\item[(iii)] If $x_* \in \C{E}_*$ then, 
\beq
E(z) := P^{(\infty)}(z) e^{-\pi i(\alpha^+_{x_*} - n\tilde{\epsilon}^+_{x_*})\theta(z)\sigma_3}(z-x_*)^{-\beta \sigma_3},
\eeq
where $\beta$ is defined in the statement of the theorem, has no jumps in a fixed neighbourhood of $x_*$. Given the behaviour of $P^{(\infty)}(z)$ as $z \rightarrow x_*$, \eqref{z to xstar}, we see that $E$ is analytic at $x_*$. Redefining $E$ by shifting the argument we obtain the result (iii).
\end{itemize}
\end{proof}

\begin{remark}
For a left edge one finds,
\beq
P^{(\infty)}(x_* + \zeta) =  E(\zeta) \sigma_3 (-\zeta)^{-\frac{\sigma_3}{4}}N \sigma_3  e^{\pi i (\alpha^-_{x_*} - n \epsilon^-_{x_*})\theta(\zeta) \sigma_3}.
\eeq
\end{remark}

\begin{lemma}
\label{xi asymptotics}

Consider $\xi(z)$ on the boundary of a disc $D_{x_*}$. The points on the boundary of the disc may be parameterised by letting $z = x_* + \zeta$ with $|\zeta|$ fixed. Let $E_l$ for $l \in \mathbb{Z}$ be functions analytic in a neighbourhood of $\zeta = 0$ which behave as $\bigO(1)$ as $n\rightarrow \infty$ uniformly for $\zeta$ in a neighbourhood of zero. We then have,
\begin{itemize}
\item[(i)] If $x_* \in \C{S}_*$ is an interior point of order $k$, then as $n \rightarrow \infty$,
\beq
\label{2 cut asymp}
n \xi(x_* + \zeta) = n\xi_+(x_*)\theta(\zeta) - \frac{i}{2k+1}\sum^\infty_{l = 0} (n^{\Delta_{x_*}} \zeta)^{2k+1-l} E_{2k+1-l}(\zeta) \theta(\zeta),
\eeq
where, for sufficiently large $n$, $E_{2k+1}(0) > 0$ and $E_{2k+1-l}(0)$ are real.
\item[(ii)] If $x_* \in \partial\C{S}_*$ is an edge point of order $k$, with $x_*$ forming the right edge of an interval in $\partial\C{S}_*$, then as $n \rightarrow \infty$,
\beq
\label{1 cut asymp}
n \xi(x_* +\zeta) = \pi i n \theta(\zeta) \epsilon^+_{x_*} - \frac{2}{2k+3}\sum^\infty_{l = 0} (n^{\Delta_{x_*}} \zeta)^{k+\frac{3}{2}-l} E_{k+1-l}(\zeta),
\eeq
where, for sufficiently large $n$, $E_{k+1-l}(0)$ are real and $E_{k+1}(0) > 0$, if $k \in 2\mathbb{N}^0$ and $E_{k+1}(0) < 0$ if $k = -1$.
\item[(iii)] If $x_* \in \C{E}_*$ is an exterior point of order $k$, then as $n \rightarrow \infty$,
\beq
\label{0 cut asymp}
n \xi(x_* +\zeta) = \pi i n \theta(\zeta) \epsilon^+_{x_*} + n \epsilon^\delta_{x_*} \log\left(n^{\Delta_{x_*}} \zeta\right) - \frac{1}{2k}\sum^\infty_{l = 0} (n^{\Delta_{x_*}} \zeta)^{2k-l} E_{2k-l}(\zeta),
\eeq
where for sufficiently large $n$, $E_{2k-l}(0)$ are real.
\end{itemize}
\end{lemma}

\begin{proof}

The proof will proceed by establishing the existence of a Laurent series representation for $\xi$ on $\partial D_{x_*}$. We will do this in two ways; firstly using the jump properties of $\xi$ and then by using \eqref{xi}. Both methods give us some, but not all, of the information required and we obtain the result by combining these properties using the uniqueness of the Laurent series. We now establish the basic form that the asymptotics of $\xi$ must take in each of the following cases,

\begin{itemize}
\item[(i)] $x_*$ is an interior point. Define $\hat{\xi}(z) := \xi(z)\theta(z)$ for $z \in D_{x_*}$. Note that for sufficiently large $n$, $\hat{\xi}$ has no jumps that intersect the boundary of the disc which together with Assumption \ref{assumption: scale appropriately} implies there exists some fixed annulus centered at $x_*$ in which $\hat{\xi}$ is analytic. Choosing $\partial D_{x_*}$ to lie within this annulus we have that on $\partial D_{x_*}$, $\hat{\xi}$ can be written as a convergent Laurent expansion centered at $x_*$. This gives,
\beq
\label{2 cut laurent}
\xi(x_* + \zeta) = \theta(\zeta) \sum^{\infty}_{j=-\infty} \xi_j \zeta^j
\eeq
where $\xi_j$ are $\zeta$-independent constants. We have yet to specify the asymptotic behaviour of $\xi_j$ as $n \rightarrow \infty$ but \eqref{2 cut laurent} implies that the large $n$ asymptotics of $\xi(x_* + \zeta)$ must have this functional dependence on $\zeta$.

\item[(ii)] $x_*$ is a right edge point. Define $\hat{\xi}(z) := \xi(z) - \pi i \epsilon^+_{x_*} \theta(z)$; this removes one of the jumps that extends to $\partial D_{x_*}$. The other jump extending to $\partial D_{x_*}$ can be removed by defining $\hat{\xi}_1(z) := \hat{\xi}(z) (z-x_*)^{-1/2}$ where the principal branch is taken with a cut on $\mathbb{R}^-$. Again, because $\partial D_{x_*}$ lies in the annulus in which $\hat{\xi}_1$ is analytic we conclude that,
\beq
\label{1 cut laurent}
\xi(x_* + \zeta) = \pi i \epsilon^+_{x_*} \theta(\zeta) + \zeta^\frac{1}{2}\sum^{\infty}_{j=-\infty} \xi_j \zeta^j.
\eeq
\item[(iii)] $x_*$ is an exterior point. Define $\hat{\xi}(z) := \xi(z) - \pi i \epsilon^+_{x_*} \theta(z)$; this removes one of the jumps that extends to $\partial D_{x_*}$. The other jump extending to $\partial D_{x_*}$ can be removed by defining $\hat{\xi}_1(z) := \hat{\xi}(z) - (\epsilon^-_{x_*} - \epsilon^+_{x_*})\log (z-x_*)$. As before we can then write,
\beq
\label{0 cut laurent}
\xi(x_* + \zeta) = \pi i \epsilon^+_{x_*} \theta(\zeta) + (\epsilon^-_{x_*} - \epsilon^+_{x_*})\log \zeta + \sum^{\infty}_{j=-\infty} \xi_j \zeta^j.
\eeq
\end{itemize}

We note at this point that we are able to rearrange \eqref{2 cut laurent}, \eqref{1 cut laurent} and \eqref{0 cut laurent} into the form \eqref{2 cut asymp}, \eqref{1 cut asymp} and \eqref{0 cut asymp} respectively. What is missing is the $n$ dependence of the coefficients in the Laurent expansion and the properties of $E_l$. To obtain these details we now use \eqref{xi}. First we write it as,
\beq
\xi(z) = \xi_\pm(p')  - \frac{1}{2} \int_{p'}^z y(s) ds, \qquad \pm \Im z > 0,
\eeq
where the integration contour does not intersect $\mathbb{R}$ and $p' = \sup(\C{R}|_{x_*})$, i.e. the right most interval edge in the disc, if $\sup(\C{R}|_{x_*})$ exists, otherwise we set $p' = x_*$. Note that when $x_*$ is of right edge or exterior type we have $\xi_\pm(p') = \pm\pi i\epsilon^+_{x_*}$, so we may write,
\beq
\xi(z) = \pi i \theta(z) \epsilon^+_{x_*} - \frac{1}{2} \int_{p'}^z y(s) ds,
\eeq
whereas when $x_*$ is an interior point we have, $\xi_\pm(p') = \pm \xi_+(p')$ so we may write,
\beq
\xi(z) = \theta(z) \xi_+(p') - \frac{1}{2} \int_{p'}^z y(s) ds.
\eeq
Now note that $y$ has a jump relation,
\beq
y_+(x) + y_-(x) = 0, \qquad x \in \C{S}.
\eeq
By mimicking the arguments used previously for $\xi$ we may delete the jumps inside $D_{x_*}$ and write,
\beq
\label{y h1 R1}
y(z) = h_1(z) \sqrt{R_1(z)} E(z) \times \begin{cases} 1, & \mbox{$x_*$ is a exterior or edge point,} \\ \theta(z),& \mbox{$x_*$ is an interior point.}  \end{cases}
\eeq
Here $h_1$ is a polynomial and $R_1$ is a rational function with only simple poles, with both functions only having poles and roots in $D_{x_*}$. The function $E(z)$ is analytic in a fixed neighbourhood of $z = x_*$ with no zeros in $D_{x_*}$. 

\begin{remark}
\label{norm rem}
We note that the factors $h_1$, $R_1$ and $E$ are only defined up to overall constants and we have some freedom how we normalise each factor. We choose $E(x_*) = 1$.
\end{remark}

In light of Assumption \ref{assumption: scale appropriately}, we have that if $h_1$ and $R_1$ have $m_h$ and $m_R$ zeros respectively and $R_1$ has $m_p$ poles we can write,
\begin{align}
h_1(x_* + n^{-\Delta_{x_*}} \sigma) &= n^{-m_h \Delta_{x_*}} \hat{h}(\sigma) \\
R_1(x_* + n^{-\Delta_{x_*}} \sigma) &= n^{-(m_R-m_p) \Delta_{x_*}} \hat{R}(\sigma)
\end{align}
where $\hat{R}(z) = \hat{r}_1(z)/\hat{r}_2(z)$ and $\hat{r}_1$, $\hat{r}_2$ and $\hat{h}$ are polynomials in which each root behaves as $\bigO(1)$ as $n\rightarrow \infty$. Note that since the behaviour of $\lim_{n \rightarrow \infty} y(z)$ near $x_*$ takes the form in Remark \ref{point types} we must have $(m_h+\frac{m_R-m_p}{2}+1)\Delta_{x_*} = 1$.

This all leads us to define,
\beq
\hat{y}_1(z) := \hat{h}(z) \sqrt{\hat{R}(z)},
\eeq
which we note is related to \eqref{scaling y} by,
\begin{align}
\lim_{n\rightarrow \infty} &n^{1 -\Delta_{x_*}}y(x_* + n^{-\Delta_{x_*}} \zeta) \times \begin{cases} 1, & \mbox{$x_*$ is a exterior or edge point,} \\ \theta(\zeta),& \mbox{$x_*$ is an interior point.}  \end{cases} \nn \\
&= \lim_{n \rightarrow \infty} \hat{y}_1(\zeta) = \hat{y}(\zeta).
\end{align}

Define $\hat{p} = n^{\Delta_{x_*}}(p'-x_*)$ and note that this also behaves as $\bigO(1)$ as $n\rightarrow \infty$. Using the taylor series for $E(z)$ about $x_*$ we then have,
\beq
\label{xi integral rep}
\xi(x_* + \zeta) = \xi_\pm(p') - \frac{1}{2}\sum^\infty_{j=0} e_j n^{-j \Delta_{x_*}} \int_{\hat{p}}^{n^{\Delta_{x_*}}\zeta} \hat{y}_1(\sigma)\sigma^j d\sigma \times \begin{cases} 1, & \mbox{$x_*$ is a exterior or edge point,} \\ \theta(\zeta),& \mbox{$x_*$ is an interior point.}  \end{cases},
\eeq
where the coefficients $e_j$ of the taylor series of $E$ have $\bigO(1)$ behaviour as $n \rightarrow \infty$.

To compute large $n$ asymptotics of the expression \eqref{xi integral rep} we need large $\zeta$ asymptotics for the function,
\beq
Y_j(\zeta) := -\frac{1}{2} \int_{\hat{p}}^{\zeta} \hat{y}_1(\sigma)\sigma^j d\sigma.
\eeq
Such asymptotics for $Y_j$ can be obtained simply by integrating the large $\zeta$ behaviour of $\hat{y}_1$ term-by-term and adding an arbitrary constant; this is valid because $\hat{y}_1(\zeta)$ can be written as a convergent series for large $\zeta$. 

We now again consider each type of critical point in turn.
\begin{itemize}
\item[(i)] When $x_*$ is an interior point, we have as $\zeta \rightarrow \infty$,
\beq
\label{int y large z} 
\hat{y}_1(\zeta) = \zeta^{2k}\sum_{l=0}^\infty \hat{q}_l \zeta^{-l},
\eeq
for $k \in \mathbb{N}^0$, where $\hat{q}_l$ are constants. We therefore have as $\zeta \rightarrow \infty$,
\begin{align}
Y_j(\zeta) = -\frac{i}{2k+1} \zeta^{2k+1+j} \sum_{l=0}^\infty q_{j,l} \zeta^{-l},
\end{align}
where $q_{j,l} := (2k+1)/(2i(2k+j-l+1)) \hat{q}_l$ are constants. Note that the above expression could have in principle contained a logarithm, however this would have been inconsistent with \eqref{2 cut laurent}. Using the above expression we obtain,
\begin{align}
n \xi(&x_*  + \zeta) =\\
&n \xi_+(x_*) \theta(\zeta)-\frac{i}{2k+1}\theta(\zeta) \sum_{l=0}^\infty (n^{\Delta_{x_*}} \zeta)^{2k+1-l} \sum^\infty_{j=0} e_j q_{j,l} \zeta^{j}\nn,
\end{align}
where we have changed $\xi_+(p')$ to $\xi_+(x_*)$ by shifting the value of $q_{2k+1,0}$. Define $E_{2k+1-l}(\zeta) := \sum^\infty_{j=0}e_j q_{j,l} \zeta^{j}$ and note that, using the expression for $q_{j,l}$, it is absolutely convergent in a neighbourhood of zero and therefore analytic there.

\item[(ii)] When $x_*$ is a right edge point, we have as $\zeta \rightarrow \infty$,
\beq
\label{edge y large z}
\hat{y}_1(\zeta) = \zeta^{k+\frac{1}{2}}\sum_{l=0}^\infty \hat{q}_l \zeta^{-l}
\eeq
where $k\in \{-1\}\cup 2\mathbb{N}^0$. We therefore have as $\zeta \rightarrow \infty$,
\begin{align}
Y_j(\zeta) = -\frac{2}{2k+3} \zeta^{k+\frac{3}{2}+j} \sum_{l=0}^\infty q_{j,l} \zeta^{-l},
\end{align}
where $q_{j,l}:=(2k+3)/(2(2k+2j-2l+3))\hat{q}_l$ are constants. Note that the above expression could have in principle contained a constant term, however this would have been inconsistent with \eqref{1 cut laurent}. Using the above expression we obtain,
\begin{align}
n \xi(&x_*  + \zeta) =\\
&n \pi i \epsilon^+_{x_*} \theta(\zeta)-\frac{2}{2k+3} \sum_{l=0}^\infty (n^{\Delta_{x_*}} \zeta)^{k+\frac{3}{2}-l} \sum^\infty_{j=0} e_j q_{j,l} \zeta^{j}\nn.
\end{align}
Define $E_{k+1-l}(\zeta) := \sum^\infty_{j=0} e_j q_{j,l} \zeta^{j}$ and note that again it defines an analytic function at zero.

\item[(iii)] When $x_*$ is an exterior point, we have as $\zeta \rightarrow \infty$,
\beq
\label{ext y large z}
\hat{y}_1(\zeta) = \zeta^{2k-1} \sum_{l=0}^\infty \hat{q}_l \zeta^{-l}
\eeq
for $k \in \mathbb{N}$. We therefore have as $\zeta \rightarrow \infty$
\begin{align}
Y_j(\zeta) =-\frac{1}{2}\hat{q}_{2k+j} \log(\zeta)- \frac{1}{2k} \zeta^{2k+j} \sum_{l=0}^\infty q_{j,l} \zeta^{-l},
\end{align} 
where $q_{j,l} = 2k/(2k+j-l)\hat{q}_l$ for $l\neq 2k+j$ and $q_{j,2k+j}$ are constants. Using the above expression we obtain,
\begin{align}
n \xi(&x_*  + \zeta) =\\
&n \pi i \epsilon^+_{x_*} \theta(\zeta) - \log(n^{\Delta_{x_*}}\zeta) \sum^\infty_{j=0} \frac{e_j \hat{q}_{2k+j}}{2n^{j\Delta_{x_*}}} -\frac{1}{2k} \sum_{l=0}^\infty (n^{\Delta_{x_*}} \zeta)^{2k-l} \sum^\infty_{j=0} e_j q_{j,l} \zeta^{j}\nn.
\end{align}
Define $E_{2k-l}(\zeta) := \sum^\infty_{j=0} e_j q_{j,l} \zeta^{j}$ which we note again is analytic at zero. Finally, defining $c := -\sum^\infty_{j=0} \frac{1}{2}e_j \hat{q}_{2k+j} n^{-j\Delta_{x_*}}$ we see that by requiring consistency with \eqref{0 cut laurent} we have $c  = n \epsilon^\delta_{x_*}$. Let us also remark that using the definition of $\epsilon$ together with $(m_h+\frac{m_R-m_p}{2}+1)\Delta_{x_*} = 1$ we see that $n\epsilon^\delta_{x_*} = \bigO(1)$ as $n \rightarrow \infty$. 
\end{itemize}

To make statements concerning the properties of $E_l$ near zero recall that by \eqref{rho h R} we have,
\begin{align}
\label{y to rho}
&y_+(x) \in i \mathbb{R}^+, \qquad \mbox{if $x\in \C{S}$,}\\
&y(x) \in \mathbb{R}, \qquad \mbox{if $x\notin \C{S}$.} \nn
\end{align}
Note that the above implies $h_1(z)$, $R_1(z)$ and $E(z)$ are real functions for $z\in D_{x_*}$. Consider,
\beq
y(x_*+\zeta) = n^{\Delta_{x_*} -1} \hat{y}_1(n^{\Delta_{x_*}} \zeta) E(x_* + \zeta)
\eeq
for $x_* + \zeta$ in a compact subset $M$ of $D_{x_*}\setminus\{x_*\}$. Note that $E(x_*) = 1$ implies, for sufficiently small $D_{x_*}$, that $E(x_* + \zeta) \in \mathbb{R}^+$ for $\zeta \in \mathbb{R} \cap M$, hence we find for sufficiently large $n$ that,
\begin{align}
\label{y to rho 2}
&\hat{y}_{1+}(n^{\Delta_{x_*}} \zeta) \in i \mathbb{R}^+, \qquad \mbox{if $x_* +\zeta\in \C{S}$,}\\
&\hat{y}_{1}(n^{\Delta_{x_*}} \zeta) \in \mathbb{R}, \qquad \mbox{if $x_* +\zeta\notin \C{S}$.} \nn
\end{align}
Letting $\zeta \in M \cap \mathbb{R}^-$ and using the large $\zeta$ expressions for $\hat{y}_1$ in each case, gives,
\begin{itemize}
\item(i) when $x_*$ is an interior point, 
\beq
\zeta^{2k}\sum_{l=0}^\infty \hat{q}_l (n^{\Delta_{x_*}} \zeta)^{-l} \in i \mathbb{R}^+,
\eeq
which implies $\hat{q}_0 \in i \mathbb{R}^+$ and $\hat{q}_l \in i \mathbb{R}$. 
\item(ii) when $x_*$ is a right edge point, 
\beq
e^{\pi i (k+\frac{1}{2})}|\zeta|^{k+\frac{1}{2}}\sum_{l=0}^\infty \hat{q}_l (n^{\Delta_{x_*}} \zeta)^{-l} \in i \mathbb{R}^+,
\eeq
which implies for $k \in 2 \mathbb{N}^0$ that $\hat{q}_0 \in \mathbb{R}^+$ and $\hat{q}_l \in \mathbb{R}$, whereas for $k = -1$ it implies $\hat{q}_0 \in \mathbb{R}^-$ and $\hat{q}_l \in \mathbb{R}$.
\item(iii) when $x_*$ is an exterior point, 
\beq
e^{\pi i (2k + 1)}|\zeta|^{2 k+1}\sum_{l=0}^\infty \hat{q}_l (n^{\Delta_{x_*}} \zeta)^{-l} \in \mathbb{R},
\eeq
which implies $\hat{q}_l \in \mathbb{R}$.
\end{itemize}
Finally, using the definition of $E_l$ in terms of $\hat{q}_l$ and combining everything gives \eqref{0 cut asymp}.
\end{proof}

\subsection{Construction of $P^{(x_*)}$}

\begin{lemma}
The local parametrix can be written in the form,
\beq \label{Px*}
P^{(x_*)}(z) = \hat{E}(z-x_*) n^{\phi \sigma_3} \widehat{\Phi}(f(z)|I, B, \{\vec{\tau}_b\}, \vec{\tau}_\infty) e^{\pi i \sigma_3 \theta(z)\alpha^+_{x_*} },
\eeq
where $\widehat{\Phi}$ is a canonical model problem with open lenses, $\hat{E}$ is an analytic function at zero, $f$ is defined in Definition \ref{f def} and $\phi \in \mathbb{R}$. The value of the $n$-dependent parameters $I$, $B$, $ \{\vec{\tau}_b\}$ and $\vec{\tau}_\infty$ are given in the proof below.
\end{lemma}
\begin{proof}

For clarity we split the construction into a number of distinct steps.

\subsubsection*{Matching the jump contours and singularities.}
First note that if $\Phi$ has a jump on a contour $C$ then $P^{(x_*)}(z)$ will have a jump on $f^{-1}(C)$. We therefore choose $I = f(\C{I}|_{x_*})$ and the lens contours $\Sigma^\gamma_\pm$ to be such that $\Sigma^L_\pm = f(\Sigma^\gamma_\pm)$ for $z \in D_{x_*}$. Similarly $P^{(x_*)}(z)$ will have singularities at $f^{-1}(b)$ for $b \in B$. To match the behaviour of $T$ we therefore must choose $B = f(\C{B}|_{x_*})$ and $\tau_{b,j} = n^{\Delta_{x_*}} t_{f^{-1}(b),j}$ for $j>0$. This choice means that $\tau_{b,j}$ have behaviour $\bigO(1)$ as $n\rightarrow \infty$ due to Assumption \ref{assumption: scale appropriately}. Note also that this choice is consistent with $P^{(x_*)}(z)$ having jumps on $\Gamma_b$ for $b \in \C{B}|_{x_*}$ since it can be verified that $f^{-1}(\Gamma_b) = \Gamma_{f^{-1}(b)}$.

\subsubsection*{Matching the jump matrices}
We now consider the jumps of $P^{(x_*)}(z)$. By direct computation we find the jumps of \eqref{Px*} will match those of $T$ if, for $z \in D_{x_*}$, we have,
\beq
\label{alpha 1}
\hat{\alpha}_\Gamma(f(z)) = \begin{cases} \alpha_\Gamma(z) - \alpha^+_{x_*} & z \in \mathbb{R} \cap D_{x_*}, \\ 
\alpha_\Gamma(z) & z \in D_{x_*} \cap \bigcup_{b\in\C{B}|_{x_*}} \Gamma_b \setminus \mathbb{R}.  \end{cases}
\eeq
By Definition \ref{alpha def} we have that for $z \in D_{x_*}$,
\beq
\alpha_{\Gamma}(z) = \begin{cases} \alpha^+_{x_*}  +\sum_{b \in \C{B}|_{x_*} } \chi_{\Gamma_{b}}(z)\alpha_{b}, & z \in \mathbb{R} \cap  D_{x_*} \\
\sum_{b \in \C{B}|_{x_*} } \chi_{\Gamma_{b}}(z) \alpha_{b}, & z \in   D_{x_*} \cap \bigcup_{b \in \C{B}|_{x_*}} \Gamma_{b} \setminus \mathbb{R},
\end{cases}
\eeq
and hence,
\beq
\hat{\alpha}_\Gamma(z) = \begin{cases} \sum_{b \in \C{B}|_{x_*} } \chi_{\Gamma_{b}}(f^{-1}(z))\alpha_{b} & z \in \mathbb{R}, \\ 
\sum_{b \in \C{B}|_{x_*} } \chi_{\Gamma_{b}}(f^{-1}(z)) \alpha_{b} & z \in \bigcup_{b\in  \C{B}|_{x_*}} \Gamma_b \setminus \mathbb{R}.  \end{cases}
\eeq
Using the fact that $\chi_{\Gamma_{b}}(f^{-1}(z)) = \chi_{\Gamma_{f(b)}}(z)$ and changing the summation dummy variable to $b \in f(\C{B}|_{x_*})$ we find that the jumps match if we choose $\hat{\alpha}_b = \alpha_{f^{-1}(b)}$ for $b \in B$. 

We now have that \eqref{Px*} fulfils conditions (a), (b) and (c) of the RH problem for $P^{(x_*)}$.

\subsubsection*{Matching the boundary conditions on $\partial D_{x_*}$}

We now turn our attention to condition (d) of the RH problem for $P^{(x_*)}$. We accomplish this by use of lemmas \ref{Pinf asymptotics} and \ref{xi asymptotics}. We now consider each case in turn:
\begin{itemize}
\item[(i)] If $x_* \in \C{S}_*$, i.e. an interior point, then for $z = x_* + \zeta \in \partial D_{x_*}$ we have,
\beq
P^{(\infty)}(z) e^{n \xi(z) \sigma_3} = E(\zeta) Q(\zeta) e^{ n\xi_+(x_*)\theta(\zeta)\sigma_3 - \frac{i}{2k+1}\sum^\infty_{l = 0} (n^{\Delta_{x_*}} \zeta)^{2k+1-l} E_l(\zeta) \theta(\zeta) \sigma_3}.
\eeq
Taking $P^{(x_*)}$ to be of the form \eqref{Px*} we have,
\begin{align}
&P^{(x_*)}(x_*+\zeta) e^{-n \xi(x_* + \zeta) \sigma_3} P^{(\infty)}(z)^{-1}= \hat{E}(z) n^{\phi \sigma_3}\times \\
& \left(I+\bigO(n^{-\Delta_{x_*}} \zeta^{-1})\right)Q(\zeta) e^{\frac{i}{2k+1}\sum^\infty_{l = 2k+2} (n^{\Delta_{x_*}} \zeta)^{2k+1-l} E_l(\zeta) \theta(\zeta) \sigma_3}Q(\zeta)^{-1} E(\zeta)^{-1}.\nn\\
&=\hat{E}(\zeta) n^{\phi \sigma_3}\left(I+\bigO(n^{-\Delta_{x_*}})\right)Q(\zeta) \left(I+\bigO(n^{-\Delta_{x_*}})\right) Q(\zeta)^{-1} E(\zeta)^{-1},\\
&=\left(I+\bigO(n^{-\Delta_{x_*}})\right).
\end{align}
In the first line of the above we have used the asymptotics for $\widehat{\Phi}(z)$ as $z \rightarrow \infty$ with $\tau_{\infty,j}= E_j(\zeta)$ and $\tau_{\infty,0} = E_{0}(\zeta) + (2k+1) ( n i \xi_+(x_*) + \pi  \alpha^+_{x_*} )$. In the second line we have set $\hat{E}(\zeta) = E(\zeta)$ and $\phi = 0$.

\item[(ii)] If $x_* \in \partial\C{S}_*$ with $x_*$ forming the right edge of an interval in $\partial\C{S}_*$, i.e. a right edge point, then for $z = x_* + \zeta \in \partial D_{x_*}$ we have,
\beq
P^{(\infty)}(z) e^{n \xi(z) \sigma_3} = E(\zeta) \zeta^{-\frac{\sigma_3}{4}}N  e^{\pi i \alpha^+_{x_*}\theta(\zeta) \sigma_3 - \frac{2}{2k+3}\sum^\infty_{l = 0} (n^{\Delta_{x_*}} \zeta)^{k+\frac{3}{2}-l} E_l(\zeta)\sigma_3}.
\eeq
Taking $P^{(x_*)}$ to be of the form \eqref{Px*} we have,
\begin{align}
&P^{(x_*)}(x_*+\zeta) e^{-n \xi(x_* + \zeta) \sigma_3} P^{(\infty)}(z)^{-1}= \hat{E}(\zeta) n^{\phi \sigma_3}\times \\
& \left(I+\bigO(n^{-\Delta_{x_*}})\right) (n^{\Delta_{x_*}} \zeta)^{-\frac{\sigma_3}{4}} N  e^{\frac{2}{2k+3}\sum^\infty_{l = k+2} (n^{\Delta_{x_*}} \zeta)^{k+\frac{3}{2}-l} E_l(\zeta)\sigma_3}
N^{-1} \zeta^{\frac{\sigma_3}{4}} E(\zeta)^{-1}.\nn\\
&=\hat{E}(\zeta) \left(I+\bigO(n^{-\frac{\Delta_{x_*}}{2}})\right)\zeta^{-\frac{\sigma_3}{4}} N\left(I+\bigO(n^{-\frac{\Delta_{x_*}}{2}})\right) N^{-1} \zeta^{\frac{\sigma_3}{4}} E(\zeta)^{-1},\\
&=\left(I+\bigO(n^{-\frac{\Delta_{x_*}}{2}})\right).
\end{align}
In the first line of the above we have set $\phi = \Delta/4$ and used the asymptotics for $\widehat{\Phi}(z)$ as $z \rightarrow \infty$ with $\tau_{\infty,j} = E_j(\zeta)$. In the second line we have set $\hat{E}(\zeta) = E(\zeta)$.

\item[(iii)] If $x_* \in \C{E}_*$ then, for $z = x_* + \zeta \in \partial D_{x_*}$ we have,
\beq
P^{(\infty)}(z) e^{n \xi(z) \sigma_3} =  E(\zeta)\zeta^{\beta \sigma_3} e^{\pi i \alpha^+_{x_*} \theta(\zeta) \sigma_3 + c \sigma_3 \log\left(n^{\Delta_{x_*}} \zeta \right) - \frac{1}{2k}\sum^\infty_{l = 0} (n^{\Delta_{x_*}} \zeta)^{2k-l} E_l(\zeta)\sigma_3}.
\eeq
Taking $P^{(x_*)}$ to be of the form \eqref{Px*} we have,
\begin{align}
&P^{(x_*)}(x_*+\zeta) e^{-n \xi(x_* + \zeta) \sigma_3} P^{(\infty)}(z)^{-1}= \hat{E}(\zeta) n^{\phi \sigma_3}\times \\
& \left(I+\bigO(n^{-\Delta_{x_*}})\right) n^{\Delta_{x_*}(\hat{\alpha}_\mathrm{tot} + \hat{c} -c) \sigma_3} \zeta^{(\hat{\alpha}_\mathrm{tot} + \hat{c} - c-\beta)\sigma_3} e^{\frac{1}{2k}\sum^\infty_{l = 2k+1} (n^{\Delta_{x_*}} \zeta)^{2k-l} E_l(\zeta)\sigma_3} E(\zeta)^{-1}.\nn\\
&= \hat{E}(\zeta) n^{\phi \sigma_3}\left(I+\bigO(n^{-\Delta_{x_*}})\right) n^{\Delta_{x_*}(\hat{\alpha}_\mathrm{tot} + \hat{c} -c) \sigma_3}
\zeta^{(\hat{c} - \C{N})\sigma_3} \left(I+\bigO(n^{-\Delta_{x_*}})\right)
 E(\zeta)^{-1},\nn\\
&=\left(I+\bigO(n^{|2\Delta_{x_*}(\hat{\alpha}_\mathrm{tot} -n \tilde{\epsilon}^\delta_{x_*})|-\Delta_{x_*}})\right).
\end{align}
In the first line of the above we have used the asymptotics for $\widehat{\Phi}(z)$ as $z \rightarrow \infty$ with $\tau_{\infty,j}(z) = E_j(z)$. In the second line we have used that $\beta := \alpha^\delta_{x_*} - n\tilde{\epsilon}^\delta_{x_*}$, $\hat{\alpha}_\mathrm{tot} = \alpha^\delta_{x_*}$,  $c = n \epsilon^\delta_{x_*}$, $n\tilde{\epsilon}^\delta_{x_*} = n \epsilon^\delta_{x_*} - \C{N}$ and $\hat{E}(z) = E(z)$. In the third line we have chosen $\phi = -\Delta_{x_*}(\hat{\alpha}_\mathrm{tot} + \hat{c} -c) = -\Delta_{x_*}(\hat{\alpha}_\mathrm{tot} -n \tilde{\epsilon}^\delta_{x_*})$. Note that in order for condition (d) to be satisfied in this case we require $|\hat{\alpha}_\mathrm{tot} -n \tilde{\epsilon}^\delta_{x_*}| < 1/2$.
\end{itemize}
Let us note that the choices made for $\tau_{\infty,j}$ in each case are consistent with the existence of a solution to the model problem due to the properties, shown in Lemma \ref{xi asymptotics}, of the $E_l$ in a small neighbourhood around zero. In particular we see that if $D_{x_*}$ is taken sufficiently small then $\vec{\tau}_{\infty}$ will be in the neighbourhood of an admissible vector and therefore the RH is solvable (see Remark \ref{admissible vec}) and \eqref{Px*} is well defined for our choice of $\vec{\tau}_\infty$.
\end{proof}

\begin{remark}
Note that we have again omitted the left edge case. The local parametrix at a left edge can be constructed in terms of the edge type model problem as follows,
\beq 
P^{(x_*)}(z) = \hat{E}(z) n^{\frac{1}{4}\Delta_{x_*} \sigma_3} \sigma_3 \widehat{\Phi}(-f(z)) \sigma_3 e^{\pi i \sigma_3 \theta(z)\alpha^-_{x_*} }.
\eeq
Checking the above satisfies the conditions of the local parametrix follows the case of $x_*$ being a right edge point.
\end{remark}

\section{Asymptotics for $K_n(x,y)$: Proof of theorem \ref{Kthm}}
We begin by expressing the correlation kernel in terms of the $Y$ RH problem,
\beq
\label{KYeqn}
K_n(x,y) = \frac{1}{2\pi i} \frac{\sqrt{w(x)w(y)}}{x-y} \begin{pmatrix}0&1\end{pmatrix} Y_+(y)^{-1}Y_+(x) \begin{pmatrix}1\\0\end{pmatrix}.
\eeq
The proof proceeds by inverting the sequence of transformations $Y \mapsto \Psi \mapsto S \mapsto T\mapsto R$ in the steepest descent analysis. The result is that for $z \in D_{x_*}$ we have,
\beq
Y(z) = e^{-\frac{n \ell \sigma_3}{2}} R(z) E(z) n^{\phi \sigma_3} \Phi(f(z)) e^{\pi i \alpha^+_{x_*} \theta(z) \sigma_3} \bar{w}(z)^{-\frac{\sigma_3}{2}}.
\eeq
Note that in the above equation we have used the fact that $ e^{\pi i \alpha^+_{x_*} \theta(z) \sigma_3} K(z)  e^{-\pi i \alpha^+_{x_*} \theta(z) \sigma_3} = \hat{K}(f(z))$. By substituting the above equation into \eqref{KYeqn} we arrive at,
\begin{align}
&K_n(x,y) = -\frac{e^{2\pi i \alpha^+_{x_*}}}{2\pi i (x-y)} \sqrt{\frac{w_\mathrm{br}(x)w_\mathrm{br}(y)}{\bar{w}_\mathrm{br}(x) \bar{w}_\mathrm{br}(y)}} \times \\
&\begin{pmatrix}-\phi_2(f(y)), & \phi_1(f(y))\end{pmatrix} n^{-\phi \sigma_3}E(y)^{-1} R(y)^{-1} R(x) E(x) n^{\phi \sigma_3} \begin{pmatrix}\phi_1(f(x)), & \phi_2(f(x))\end{pmatrix}^T \nn,
\end{align}
where we have introduced the functions $\phi_i(z)$ from Definition \eqref{phi def}. Noting that,
\beq
\frac{w_\mathrm{br}(x)}{\bar{w}_\mathrm{br}(x)} = e^{-2\pi i \alpha_\Gamma(x)},
\eeq
we obtain,
\begin{align}
&K_n(x,y) = -\frac{e^{\pi i (2 \alpha^+_{x_*} - \alpha_\Gamma(x) - \alpha_\Gamma(y))}}{2\pi i (x-y)}  \times \\
&\begin{pmatrix}-\phi_2(f(y)), & \phi_1(f(y))\end{pmatrix} n^{-\phi \sigma_3}E(y)^{-1} R(y)^{-1} R(x) E(x) n^{\phi \sigma_3} \begin{pmatrix}\phi_1(f(x)), & \phi_2(f(x))\end{pmatrix}^T. \nn
\end{align}
Next, using the properties of $R$, we have that,
\beq
R(x_* + n^{-\Delta_{x_*}} v)^{-1} R(x_* + n^{-\Delta_{x_*}} u) = I + \bigO(\frac{u-v}{n^{\Delta_{x_*}}}).
\eeq
Similarly we have that,
\beq
E(x_* + n^{-\Delta_{x_*}} v)^{-1} E(x_* + n^{-\Delta_{x_*}} u) = I + \bigO(\frac{u-v}{n^{\Delta_{x_*}}}).
\eeq
Using the above expression together with those for $\bar{w}_\mathrm{br}$ we obtain,
\begin{align}
K_n &(x_* + n^{-\Delta_{x_*}} u, x_* + n^{-\Delta_{x_*}}v) = -n^{\Delta_{x_*}}e^{\pi i (2 \alpha^+_{x_*} - \alpha_\Gamma(x_* + n^{-\Delta_{x_*}} u) - \alpha_\Gamma(x_* + n^{-\Delta_{x_*}} v))}  \times \nn \\
&\frac{\phi_1(u)\phi_2(v) - \phi_1(v)\phi_2(u)}{2\pi i (u-v)}(1 + \bigO(n^{-\Delta_{x_*} + 2|\phi|})),
\end{align}
as $n \rightarrow \infty$. Finally we note,
\begin{align}
\alpha_\Gamma(x_* + n^{-\Delta_{x_*}} u) &= \sum_{b \in \C{B}} \chi_b(\Re (x_* + n^{-\Delta_{x_*}} u)) \alpha_b \nn \\
&=  \sum_{b \in \C{B}|_{x_*}} \chi_b(\Re (x_* + n^{-\Delta_{x_*}} u)) \alpha_b + \sum_{b \notin \C{B}|_{x_*}} \chi_b(\Re (x_* + n^{-\Delta_{x_*}} u)) \alpha_b \nn\\
&= \alpha^+_{x_*} + \sum_{b \in \C{B}|_{x_*}} \chi_b(\Re (x_* + n^{-\Delta_{x_*}} u)) \alpha_b \nn \\
&= \alpha^+_{x_*} + \sum_{b \in B} \chi_b(\Re u) \alpha_b.
\end{align}
We therefore obtain,
\begin{align}
n^{-\Delta_{x_*}} K_n &(x_* + n^{-\Delta_{x_*}} u, x_* + n^{-\Delta_{x_*}}v) = \\
&-\frac{e^{-\pi i (\hat{\alpha}_\Gamma(u) + \hat{\alpha}_\Gamma(v))}}{2\pi i } \frac{\phi_1(u)\phi_2(v) - \phi_1(v)\phi_2(u)}{u-v}(1 + \bigO(n^{-\Delta_{x_*} + 2|\phi|}))\nn,
\end{align}
from which the theorem follows. The data for the model problem comes directly from the assignments made in the construction of the local parametrix. This is easy in the case of $I$, $B$ and $\{\vec{\tau}_b\}$. For $\vec{\tau}_\infty$ note that from the construction of the local parametrix in the exterior and edge case we have that at finite $n$, $\tau_{\infty,l} = E_l(n^{-\Delta_{x_*}} f(z))$. In the case of an interior point we have $\tau_{\infty,l}(z) = E_l(n^{-\Delta_{x_*}} f(z))$ with the exception of $\tau_{\infty,0}$ for which we have,
\beq 
\tau_{\infty,0}(z) = E_{0}( n^{-\Delta_{x_*}} f(z)) + (2k+1) ( n i \xi_+(x_*) + \pi  \alpha^+_{x_*} ).
\eeq
The vector $\vec{\tau}_\infty$ appearing in the model problem for the limiting kernel will be given by $\lim_{n\rightarrow \infty} \vec{\tau}_\infty(x_* + n^{-\Delta_{x_*}} u)$, which we now compute.

For the edge and exterior case we have,
\beq
\label{lim tau inf}
\lim_{n\rightarrow \infty} \vec{\tau}_\infty(x_* + n^{-\Delta_{x_*}} u) = \lim_{n\rightarrow \infty} E_l(0)
\eeq
and for the interior case we have \eqref{lim tau inf} for $l \neq 0$ and,
\beq
\lim_{n\rightarrow \infty} \tau_{\infty,0}(x_* + n^{-\Delta_{x_*}} u) = \lim_{n\rightarrow \infty}  (E_{0}(0) + (2k+1) ( n i \xi_+(x_*) + \pi  \alpha^+_{x_*} )).
\eeq
Now note that,
\beq
\hat{\xi}(\zeta) = \lim_{n \rightarrow \infty} Y_0(\zeta)
\eeq
which, using $E_l(0) = q_{0,l}$, can be written for large $\zeta$ as follows,
\begin{itemize}
\item[(i)] if $x_*$ is an interior point of order $k$ we have,
\beq
\hat{\xi}(\zeta)  = -\frac{i}{2k+1} \sum^{2k+1}_{j = 0} E_j(0)\zeta^j + \bigO(\zeta^{-1}),
\eeq
\item[(ii)] if $x_*$ is an edge point of order $k$ we have,
\beq
\zeta^{-\frac{1}{2}}\hat{\xi}(\zeta)  =  -\frac{2}{2k+3} \sum^{k+1}_{j = 0} E_j(0) \zeta^j +\bigO(\zeta^{-1}),
\eeq
\item[(iii)] if $x_*$ is an exterior point of order $k$ we have,
\beq
\hat{\xi}(\zeta) = - \hat{c} \log\zeta - \frac{1}{2k}\sum^{2k}_{j = 0} E_j(0) \zeta^j + \bigO(\zeta^{-1}).
\eeq
\end{itemize}
This completes the final part of the proof showing that the model problem in Definition \ref{model problem at xs} is the one appearing in Theorem \ref{Kthm}.
\qed

\section{Example applications}
\subsection{The Painlev\'{e} II kernel in the bulk}
We consider the model,
\beq
\frac{1}{Z_n} e^{-n \Tr V(M)} dM,
\eeq
where $V(M) = \frac{t}{2} M^2 + \frac{1}{4}M^4$. This model possess two distinct phases; for $t > -2$ the equilibrium measure support is connected, whereas for $ t< -2$ the support consists of two disjoint intervals. The phase transition between these two phases was first analysed in \cite{2002PIIBI}, in which it was shown that as the two cuts meet at the origin, the kernel there may be expressed in terms of Painlev\'{e} II transcendents. We now will see how this result follows from Theorem \ref{Kthm}.

From \cite{2002PIIBI} we have the spectral curve,
\beq
y(x) = (x^2 + 2c)\sqrt{x^2 - a^2}
\eeq
with $a = (\frac{1}{3}(-2t + 2(t^2 + 12)^\frac{1}{2}))^\frac{1}{2}$ and $c=\frac{1}{3}(t + (\frac{1}{4}t^2 + 3)^\frac{1}{2})$.
Consider the spectral curve when $t = -2$,
\beq
y(x) = x^2\sqrt{x^2 - 4}.
\eeq
Now let $t \rightarrow -2$ as $n \rightarrow \infty$. In this limit we have $\C{S}_* = [-2,2]$ and $\C{H}_* = \{0\}$. Using Defintion \ref{hatH} we therefore see that $\hat{\C{H}}_* = \{0\}$. So $x_* = 0$ is in $\C{P}_*$ and from Definition \ref{point types} is an interior point of order $k = 1$. From Definition \ref{Delta def} we obtain $\Delta_{x_*} = \frac{1}{3}$. For this model we have $\C{I} = \mathbb{R}$ and $\C{B} = \varnothing$ which implies, using Definition \ref{model problem at xs}, $I = \mathbb{R}$ and $B = \varnothing$.

Now consider the limit $n \rightarrow \infty$ with $t=-2 + n^{-\frac{2}{3}}\tau$, this gives $c = \frac{1}{4} \tau n^{-\frac{2}{3}}$ and we note that this means $\sqrt{c}$ scales appropriately. Using Definition \ref{scaling y def} we therefore have,
\beq
\hat{y}(\zeta) = i(2\zeta^2 + \tau)
\eeq
and
\beq
\hat{\xi}(\zeta) = -i(\frac{1}{3}\zeta^3 + \tau \zeta).
\eeq
Hence, from Definition \ref{model problem at xs}, we have $\tau_{\infty,3} = 1$, $\tau_{\infty,2} = 0$ and $\tau_{\infty,1} = 3\tau$. This completes the construction of the data for the canonical model problem. Note that
this model problem is related, up to a rescaling of $z$, to the standard model problem for PII by \eqref{PIIrel}. Indeed we have,
$\Psi^{PII}(\bar{c} z) =  e^\frac{i \tau_{\infty,0}}{3}\widehat{\Phi}(z)$ for $z$ in the upper half plane and some constant $\bar{c}$. Note that the factor $e^\frac{i \tau_{\infty,0}}{3}$ cancels out in the expression for the kernel.

From Theorem \ref{Kthm} we have that there exists a constant $c$ such that,
\beq
\lim_{n\to\infty}c n^{-\frac{1}{3}}K_n\left(n^{-\frac{1}{3}} c u, n^{-\frac{1}{3}} c v\right) = -\frac{\phi_1(u)\phi_2(v) - \phi_1(v)\phi_2(u)}{2 \pi i (u-v)},
\eeq
where,
\beq
\begin{pmatrix} \phi_1(z) \\ \phi_2(z)\end{pmatrix} := \Psi^{PII}(z) \begin{pmatrix}1 \\ 0 \end{pmatrix},
\eeq
for $z$ outside the lens.

\subsection{A new kernel at a hard edge}
Consider the model,
\beq
\frac{1}{Z_n} |\det M|^{2\alpha_1} |\det(M-t)|^{2\alpha_2} e^{n M} dM,
\eeq
where $M$ is a negative definite hermitian matrix. We are going to consider the regime close to the origin
as we allow $t \rightarrow 0$ as $n \rightarrow \infty$. The matrix is taken to be negative definite so that the hard edge is a right edge point, which simplifies the construction.

The spectral curve is the standard Marchenko-Pastur curve associated with the Laguerre unitary ensemble,
\beq
y(x) = -\sqrt{\frac{x+4}{x}}.
\eeq
From the expression for $y$ we have $\C{S}_* = [-4,0]$ and $\C{H}_* = \{0\}$. Using Defintion \ref{hatH} we see that $\hat{\C{H}}_* = \{0\}$ and therefore that the origin is a hard edge point with $k = -1$. From Definition \ref{Delta def} we obtain $\Delta_{x_*} = 2$. For this model we have $\C{I} = \mathbb{R}^{-}$ and $\C{B} = \{0, t\}$ with $\alpha_0 = \alpha_1$ and $\alpha_t = \alpha_2$. We require $t$ scales appropriately and therefore we set $t = n^{-2} \tau$. This implies we have $I = \mathbb{R}^-$ and $B = \{0,\tau\}$.
The scaling limit of $y$ is,
\beq
\hat{y}(\zeta) = -2\zeta^{-\frac{1}{2}},
\eeq
and therefore,
\beq
\hat{\xi} = 2\zeta^\frac{1}{2},
\eeq
which gives $\tau_{\infty,0} = -1$.

This completes the construction of the data for the canonical model problem. To summarise, the model problem that will appear in the kernel is,

\begin{figure}[t]
\centering 
\includegraphics[scale=0.4]{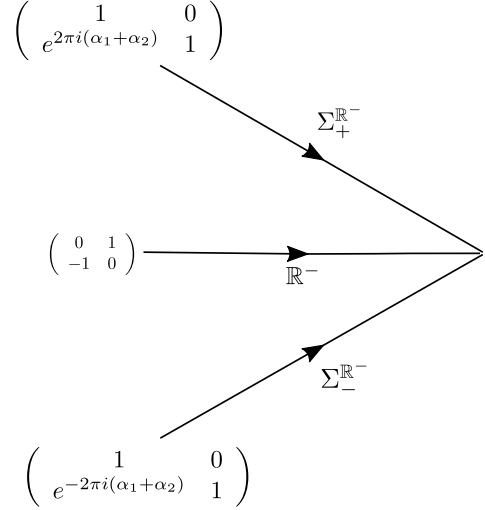}
\caption{The jump contours for the RH problem from which the new kernel at the hard-edge is constructed.}
\label{PVjumps1}
\end{figure}

\subsubsection*{RH for $\widehat{\Phi}$}
\begin{itemize}
\item[(a)] $\widehat{\Phi} : \mathbb{C} \setminus \mathbb{R}^- \setminus \Sigma^{\mathbb{R}^-}_+\setminus \Sigma^{\mathbb{R}^-}_- \rightarrow \mathbb{C}^{2 \times 2}$ is analytic in $z$.
\item[(b)] $\widehat{\Phi}$ has the jumps shown in Figure \ref{PVjumps1}
\item[(c)] As $z \rightarrow \infty$,
\beq
\widehat{\Phi}(z)=\left(I+\bigO(z^{-1})\right) z^{-\frac{\sigma_3}{4}} N e^{2 z^\frac{1}{2}\sigma_3}.
\eeq
\item[(d)] As $z \rightarrow 0$,
\beq
\widehat{\Phi}(z) = \bigO(1)(z-b)^{\frac{\alpha_1}{2}  \sigma_3} \hat{K}(z).
\eeq
This can be written as,
\beq
\widehat{\Phi}(z)=\bigO \begin{pmatrix}z^\frac{\alpha_1}{2} & z^{-\frac{\alpha_1}{2}}\\  z^\frac{\alpha_1}{2} & z^{-\frac{\alpha_1}{2}} \end{pmatrix},
\eeq
for $z$ outside the lens and
\beq
\widehat{\Phi}(z)=\bigO \begin{pmatrix}z^{-\frac{\alpha_1}{2}} & z^{-\frac{\alpha_1}{2}}\\  z^{-\frac{\alpha_1}{2}} & z^{-\frac{\alpha_1}{2}} \end{pmatrix},
\eeq
for $z$ inside the lens.

As $z \rightarrow \tau$,
\beq
\widehat{\Phi}(z)=\bigO \begin{pmatrix}(z-\tau)^{-\frac{\alpha_2}{2}} & (z-\tau)^{-\frac{\alpha_2}{2}}\\  (z-\tau)^{-\frac{\alpha_2}{2}} & (z-\tau)^{-\frac{\alpha_2}{2}} \end{pmatrix}.
\eeq
\end{itemize}
Let us remark that this RH problem leads to Lax pairs which produce a second order ODE and therefore
will be one of the standard Painlev\'{e} equations.

From Theorem \ref{Kthm} we have,
\beq
\lim_{n\to\infty}n^{-\frac{1}{3}}K_n\left(n^{-\frac{1}{3}} u, n^{-\frac{1}{3}} v\right) = -\frac{\phi_1(u)\phi_2(v) - \phi_1(v)\phi_2(u)}{2 \pi i (u-v)},
\eeq
where,
\beq
\begin{pmatrix} \phi_1(z) \\ \phi_2(z)\end{pmatrix} := \Phi(z) \begin{pmatrix}1 \\ 0 \end{pmatrix},
\eeq
for $z \in \mathbb{H} \setminus I \setminus \cup_{b \in B} \Gamma_{b}$.

\bibliographystyle{hieeetr}
\bibliography{../researchbib}

\end{document}